\documentclass[onecolumn,11pt,draftclsnofoot]{IEEEtran}
\usepackage{graphicx,cite}
\usepackage{epsfig}
\usepackage[cmex10]{amsmath}
\usepackage{setspace}
\usepackage{amssymb}
\usepackage{array}
\usepackage{fixltx2e}
\usepackage{amsfonts}
\usepackage{lineno}
\usepackage{epstopdf}
\usepackage[center]{caption}
\usepackage{subcaption}
\usepackage{algorithm}
\usepackage{algorithmic}
\usepackage{color}

\newtheorem{proposition}{Proposition}

\usepackage{multirow}
\usepackage{booktabs}

\begin{document}

\title{Optimization of Free Space Optical Wireless Network for Cellular Backhauling}

\author{\IEEEauthorblockN{Yuan Li\IEEEauthorrefmark{1}, Nikolaos Pappas\IEEEauthorrefmark{3}, Vangelis Angelakis\IEEEauthorrefmark{3}, Micha{\l} Pi\'{o}ro\IEEEauthorrefmark{1}\IEEEauthorrefmark{2}, and Di Yuan\IEEEauthorrefmark{3}}
\IEEEauthorblockA{\IEEEauthorrefmark{1}Department of Electrical and Information Technology, Lund University, Lund, Sweden}
\IEEEauthorblockA{\IEEEauthorrefmark{2}Institute of Telecommunications, Warsaw University of Technology, Warsaw, Poland}
\IEEEauthorblockA{\IEEEauthorrefmark{3}Department of Science and Technology, Link\"{o}ping University, Norrk\"{o}ping, Sweden \\
Email: yuan.li@eit.lth.se, \{nikolaos.pappas,vangelis.angelakis\}@liu.se, mpp@tele.pw.edu.pl, di.yuan@liu.se}

\thanks{The authors would like to acknowledge networking support by the COST Action IC1101 OPTICWISE (Optical Wireless Communications - An Emerging Technology). The work of M.~Pi\'{o}ro was supported by National Science Center (Poland) under grant 2011/01/B/ST7/02967.}
}

\maketitle

\begin{abstract}

With densification of nodes in cellular networks, free space optic (FSO) connections are becoming an appealing low cost and high rate alternative to copper and fiber as the backhaul solution for wireless communication systems. To ensure a reliable cellular backhaul, provisions for redundant, disjoint paths between the nodes must be made in the design phase. This paper aims at finding a cost-effective solution to upgrade the cellular backhaul with pre-deployed optical fibers using FSO links and mirror components. Since the quality of the FSO links depends on several factors, such as transmission distance, power, and weather conditions, we adopt an elaborate formulation to calculate link reliability. We present a novel integer linear programming model to approach optimal FSO backhaul design, guaranteeing $K$-disjoint paths connecting each node pair. Next, we derive a column generation method to a path-oriented mathematical formulation. Applying the method in a sequential manner enables high computational scalability. We use realistic scenarios to demonstrate our approaches efficiently provide optimal or near-optimal solutions, and thereby allow for accurately dealing with the trade-off between cost and reliability.

\end{abstract}

\IEEEpeerreviewmaketitle

\section{Introduction}

A cellular backhaul comprises the connections between base stations and the core network components (e.g., radio network controller and base station controller). With high-speed data services and network densification using small cells, the need of upgrading the backhaul and increasing its capacity is rapidly growing \cite{Ford2013}.
At present, three transport media are primarily used for backhaul solutions: copper (about $90\%$), microwave radio links (about $6\%$), and optical fibers (about $4\%$)~\cite{Tipmongkolsilp2011}. Leased copper lines are becoming an infeasible options for meeting future backhaul demands, as the data rate is low while the price increases linearly with capacity. On the other hand, optical fiber can support very high data rates but need substantial initial investment. Conventional radio-frequency (RF) technologies have been widely studied for backhauling, and the RF mesh networking paradigm has attracted significant attention for backhaul topology. However, RF links have rather limited data rates, and are prone to interference and security problems. As frequencies go up with millimeter waves, RF transmissions are also hampered by distance and even weather conditions~\cite{Frey99}. What is more, the licensed part of the spectrum comes at additional costs \cite{Chia2009}.

Free space optics (FSO)-based networking is an attractive alternative for the next generation cellular networks~\cite{Demers2011}. An FSO link uses the free space between a pair of laser-photodetector transceivers to transport data. The FSO beam has a wavelength in the micrometer range, yielding advantages in terms of free license, interference immunity, and high bandwidth, among others. Such FSO wireless transceivers are already commercially available~\cite{Chan06}, and can be deployed to establish optical links to support several gigabits per second over a distance of a few kilometers with fast deployment. FSO links are considerabley more cost-efficient than optical fibers. As they operate with similar data rates, FSO and optical fibers are easily combined together in a network \cite{Refai2006}.

With the advantages it brings, the FSO technology is a good complementary option for traditional radio-based wireless technologies. However, there are also inherent difficulties for the deployment of an FSO-based system. FSO links can be set up only when both end nodes are in line-of-sight, and the link availability is sensitive to weather conditions, such as fog and precipitation~\cite{Ghassemlooy2012}. Thus the aspect of reliability is of high significant in deploying an FSO-based network.

In this paper, we address the network optimization problem of upgrading a fiber optical backhaul with FSO technology. Since optical fibers are insensitive to weather conditions, any pair of nodes connected by an optical fiber is treated as reliable. For FSO links, on the other hand, reliability is explicitly considered in our system model. A specific aspect in the network design problem we consider is the possible use of mirrors. When two FSO nodes are not in line-of-sight, mirrors installed at a third node may potentially provide connectivity. An extension is to connect two FSO nodes with a sequence of mirrors. However, as mirrors do not provide any amplification, the total distance of such a mirror path is subject to a limit.

The network design optimization problem consists in determining the locations of setting up FSO links and mirror links to provide connectivity, possibly via multiple hops, between all FSO nodes with traffic demand. A specific design concern is to guarantee a desired degree of resilience to cope. That is, the FSO nodes remain connected even if some of the wireless optical links become unavailable. This, so called network survivability, is formulated using $K$-connectivity, i.e., there are (at least) $K$ link-disjoint paths connecting each pair of FSO nodes, where $K$ is a parameter. A network topology example with $2$-connectivity is shown in Fig. \ref{fig:scenario}.

\begin{figure*}
\centering
\includegraphics[width=4in]{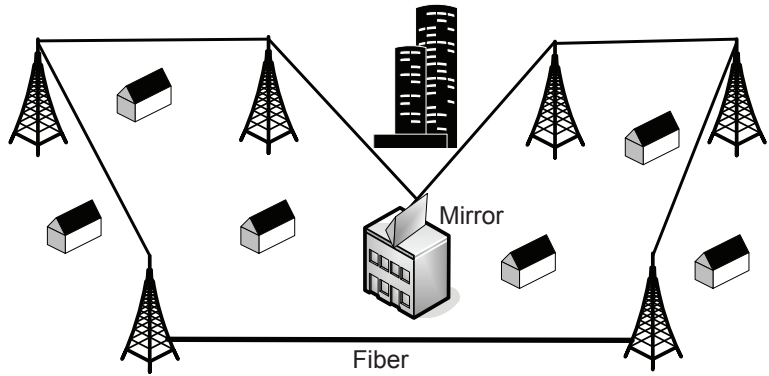}
\caption{A backhaul network deployment scenario with FSO transceivers, mirrors and fibers}\label{fig:scenario}
\end{figure*}

The performance objectives are deployment cost and network reliability. The former ranges from the cost of FSO transceivers, mirrors, to that of leasing building roofs for installing mirrors, whereas the latter is dependent on link lengths and the number of links. Considering these aspect along with $K$-connectivity lead to a comprehensive system model. We present following contributions to the FSO network optimization problem. First, we derive a novel integer linear programming approach that deals with the cost, reliability, and $K$-connectivity aspects using FSO and mirror links. The approach guarantees global optimality for up to medium-size planning scenarios. Next, for better scalability when dealing with large-size scenarios, we develop an alternative, path-oriented optimization formulation. This formulation allows for problem decomposition, by making use of a column generation method that finds candidate paths for each FSO node pair in an incremental manner. Our second solution approach consists in applying the column generation method sequentially to the FSO node pairs, to enable obtaining near-optimal solutions time-efficiently. We report computational results for realistic planning scenarios of FSO backhaul design, to demonstrate the viability of the optimization framework as well as to shed light on the trade-off between cost and reliability.

The remainder of the section is organized as follows. Related works are reviewed in Section \ref{sec:relate}.  Section \ref{sec:system} introduces the FSO channel model and basic notations. The exact integer programming model is presented in Section \ref{sec:formulation}, and then in Section \ref{sec:approach}, the sequential computation approach is elaborated. The numerical study illustrating optimal solutions and comparing algorithms is conducted in Section \ref{sec:numerical}. Finally, Section \ref{sec:conclusion} concludes the paper.

\section{Related work} \label{sec:relate}
In the survivable network design there is a significant body of works. Various mathematical programming models have been proposed for designing $K$-connected graphs, see for example
\cite{Kerivin2005} and references therein. Papers \cite{Khandekar2013,Bendali2010,Botton2013} study the problem of $K$-connected network design with node degree constraint or limiting the number of hops. An algorithm for survivable network design with reliable links, which can be shared by link-disjoint paths is proposed in \cite{Zotkiewicz2010}.

In our paper, we consider the design of a $K$-connected graph with both reliable links (optical fibers) and mirror links. Furthermore, we constrain the total length of consecutive mirror links, which may only be part of a path connecting a pair of FSO transceivers. The previous constraint differentiates our work from the traditional problem which limits the number of the hops per path.

With FSO being an affordable alternative to optical fibers for high data-rate communication,  there is a vast amount of research focused on FSO networking. The majority of works studies the factors (weather, alignment, turbulence, etc.) that affect the link performance and also investigates different techniques (modulation and coding schemes, electronics design, etc.) to improve the FSO link quality (see for example \cite{Borah2012} and the references therein). On the other hand, work that focuses on network optimization aspects on FSO communications are rather limited. Recently though the problem of network topology design has attracted attention.

In \cite{Llorca2004, Zhuang2004, Cao2008} simple topology design models have been studied for the FSO networks. The work \cite{Llorca2004} develops heuristics to design $2$-degree topologies ($2$ transceivers per node) and $3$-degree topologies ($3$ transceivers per node) with a minimum number of links. Paper \cite{Zhuang2004} adopts a combination of multiple objectives, such as minimizing the cost in the physical layer and the congestion in the logical layer to generate a $2$-degree topology heuristically. The paper \cite{Cao2008} presents an integer programming model to maximize the network throughput by installing as many as possible FSO links when the number of optical transceivers per node is limited. Similarly \cite{Abhish2007} proves the problem to be NP-hard and proposes algorithms for integrated topology control
and single-path or multi-path routing. In our work, we do not restrict the number of FSO transceivers in each node since in practice, as FSO transceivers can be collocated (being mounted at different heights) in the same node.

Integer programming models for designing topologies based on FSO have also been proposed. In \cite{Son2010} the backbone of wireless mesh networks with FSO links is designed by maximizing the so called algebraic connectivity. However, the algebraic connectivity cannot capture the resilience of connections for each pair of nodes.  The same objective is also studied in paper \cite{Zhou2013}.
The \cite{Ouveysi2010} presents a model that performs load balancing and provides link-disjoint paths for each pair of nodes. However, the proposed simple algorithm for finding link-disjoint paths does not take into account that finding shortest path one-by-one may not produce feasible link-disjoint paths (see appendix C.4.1 in \cite{pioro2004}).
Joint topology design and load balancing in FSO networks is addressed in \cite{Son2014} where the reformulation linearization technique (RLT) is applied to obtain linear programming (LP) relaxations of the original complex problem, and then incorporate the LP relaxations into a branch-and-bound framework.

\section{System model}\label{sec:system}

\subsection{Link Model}

FSO is an optical communication technology that uses light propagating in the free space to transmit data using a laser beam, in a wireless manner between two FSO transceivers. During propagation, the beam is attenuated due to photon absorption and scattering. Thus, the propagation loss of FSO links is affected by rain, fog, wind, temperature, pointing errors, etc. \cite{Ghassemlooy2012}. Different channel models have been proposed considering the impact of the atmospheric condition. A widely used model is the log-normal channel model (the amplitude fluctuation is log-normal), which is suitable for weak atmospheric turbulence in clear weather condition. The probability density function of the received irradiance $I$ is $p(I) = \frac{1}{2 \sqrt{2\pi} I \sigma_X} \mbox{exp} \{-\frac{(\mbox{ln} {I} - \mbox{ln} {I_0})^2}{8\sigma^2_X}\}$ where $\sigma_X$ is covariance of the log-amplitude fluctuation $X$, $I_0$ is the average irradiance when there is no turbulence in the channel \cite{Zhu2002}. The standard deviation $\sigma^2_X$ is approximated by $\sigma^2_X = 0.30545 (\frac{2 \pi}{\lambda})^{7/6} C^2_n (h) l ^{11/6}$, where $\lambda$ is the wavelength, $l$ is the transmission distance, and $C^2_n (h)$ is the so called \emph{index of refraction structure parameter} with a constant altitude $h$, which expresses the strength of the atmospheric turbulence.

The reliability of a FSO link with length $l$ is defined as the cumulative probability of the irradiance above a threshold of the received signal intensity $I_{th}$. It can be computed as in equation \eqref{channel} according to \cite{Son2010} where $erf(\cdot)$ is the error function. An FSO link can be established when its reliability is bigger than a given a reliability threshold $\Gamma_{th}$. Note that the value of the reliability of a link take values between $0$ and $1$.
\begin{equation}\label{channel}
\Gamma(l) = \int_{I_{th}}^{\infty} p(I) dI  = \frac{1}{2} - \frac{1}{2} erf \left(\frac{\mbox{ln}(I_{th}/I_0)}{2 \sigma_X \sqrt{2}} \right).
\end{equation}

For an atmospheric channel close to the ground, e.g., $h < 18.5$ m,  $C^2_n (h)$ ranges from $10^{-13}$ to $10^{-17}$ $m^{-2/3}$ for a strong to weak atmospheric turbulence. For a fixed ratio $I_{th}/I_0$ and constant weather turbulence, the reliability of the link depends on the transmission distance, and the longer the link, the smaller the reliability.

\subsection{Network Model}
For the problem considered in this paper we assume that an initial (undirected) network graph is given, including a set of nodes $\mathcal V$ and a set of potential links $\mathcal E$ (see Table~\ref{tab:notations}).
The set of nodes consists of the set of the base stations for deploying FSO transceivers (FSO nodes, $\mathcal V^F$), the buildings (mirror nodes, $\mathcal V^M$) on which mirrors can be deployed, i.e., $\mathcal V = \mathcal V^F \cup \mathcal V^M$ where $\mathcal V^F, \mathcal V^M$ are disjoint. There are three types of links: fiber links, FSO links, and mirror links. The set of all fiber links is denoted by $\mathcal E^O$, the set of all potential FSO links -- by $\mathcal E^F$, and the set of all potential mirror links -- by $\mathcal E^M$. The set of links is thus defined as $\mathcal E = \mathcal E^O \cup \mathcal E^F \cup \mathcal E^M$ where the sets $\mathcal E^O, \mathcal E^F$ and $\mathcal E^M$ are disjoint. Links are undirected and the end nodes of link $e \in \mathcal E$ are denoted by $p_e,q_e \in \mathcal V$.

The fiber links constituting set $\mathcal E^O$ are just links between FSO nodes established over the fibers and are assumed to be already deployed in the considered backhaul network.
Each FSO link $e \in \mathcal E^F$ connects two FSO nodes $p_e$ and $q_e$ and is composed of parallel mirror paths between $p_e$ and $q_e$. Each such path is composed of a sequence (possibly empty) of mirror links and is equipped with two FSO transceivers placed in its end nodes ($p_e$ and $q_e$), and the mirrors in the transit (mirror) nodes. When the mirror path is empty then two FSO transceivers of link $e$ are connected by a direct laser beam (the end nodes must be in the line of sight). Note that each mirror path of an FSO link (including the empty path) can be used several times and then each copy is equipped with its own FSO transceivers and mirrors. In a mirror link $e \in \mathcal E^M$ at least one end node is a mirror node and the other is either another mirror node or an FSO node. The FSO links with the end nodes not in the line of sight are represented by the subset $\mathcal E^N  \subseteq  \mathcal E^F$.

The sets of bi-directed arcs corresponding to the sets of undirected fiber links, FSO and mirror links are denoted by $\mathcal A^O, \mathcal A^F, \mathcal A^M$, respectively. Each undirected link $e \in \mathcal E$ corresponds two oppositely directed arcs $a', a'' \in \mathcal A$. We define a mapping $f: \mathcal A \rightarrow \mathcal E$ such that $f(a')= e$, $f(a'')=e$ and $f^{-1}(e)=\{a',a''\}$.
We use $\delta^+_{av}$ to represent whether arc $a$ is outgoing from node $v$, $\delta^+_{av}=1$, or not, $\delta^+_{av}=0$. $\delta^-_{av}$ indicates whether arc $a$ is incoming to node $v$, $\delta^-_{av}=1$, or not, $\delta^-_{av}=0$.

Two paths with the same end nodes are called link-disjoint when they have no FSO link nor mirror link in common, but may contain the common fiber links. We aim at constructing a subgraph that assures $K$ link-disjoint paths for each pair of FSO nodes, meaning that the subgraph is $K$-connected in this sense.
We formulate the problem by utilizing a multi-commodity flow model. Each pair of FSO nodes is treated as a commodity; the set of all such commodities is denoted by $\mathcal D$. The source and the destination of a commodity $d$ are represented by $s_d \in \mathcal E^F$ and $t_d \in \mathcal E^F$, respectively.

\begin{figure}[!htbp]
\centering
\includegraphics[height=1.3in,width=5in]{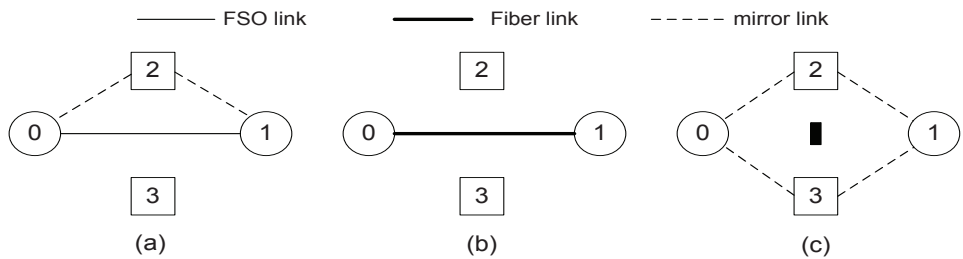}
\caption{An example of 2-connected pairs}
\label{fig:disjointpaths}
\end{figure}

In the considered model, FSO transceivers, fibers and mirrors are taken into consideration. They play different roles in the construction of the $K$-connected graph. A direct FSO link or a mirror link can be established only when there is line of sight between their end-nodes. The fiber links are assumed to already exist and we just use them for designing the paths.
For a given commodity, the paths sharing a fiber link are treated as link-disjoint, since fibers are much more reliable than the FSO and mirror links. The mirror nodes are used to connect FSO nodes which are not in the line of sight, or to create parallel paths to assure path diversity. Note that not all mirror nodes will in general be used in the optimal solution. When a mirror path traverses a mirror node, it is assigned its own mirror at that node.

Fig. \ref{fig:disjointpaths} (where circles, squares and black rectangles represent FSO nodes, mirror nodes and obstacles, respectively) illustrates three cases of link-disjoint paths between nodes $0$ and $1$: (a) one path consists of mirrors and the other path consists of an FSO link; (b) an optical fiber can be treated as a pair of link-disjoint paths; (c) two link-disjoint paths consists of sole mirror links due to an obstacle between node $0$ and node $1$.

\begin{table}
\centering
\setlength{\tabcolsep}{1pt}
\caption{Notations}
\label{tab:notations}
\begin{tabular}{lp{16cm}}%
\cline{1-2}
\hline
$\mathcal E^{F}$ & the set of FSO links\\
$\mathcal E^{O}$ & the set of fiber links\\
$\mathcal E^{M}$ & the set of mirror links\\
$\mathcal E$ & the set of all system links ($\mathcal E= \mathcal E^F \cup \mathcal E^O  \cup \mathcal E^M$)\\
$\mathcal E^{N}$ & the set of links with end nodes not in line of sight ($\mathcal E^{N} \subseteq \mathcal E^F$)\\
$\mathcal A^{O}$ & the set of fiber arcs\\
$\mathcal A^{F}$ & the set of FSO arcs\\
$\mathcal A^{M}$ & the set of mirror arcs\\
$\mathcal A$ & the set of system arcs, $\mathcal A= \mathcal A^{F} \cup \mathcal A^{O} \cup \mathcal A^M$\\
$\mathcal V^F$ & the set of FSO nodes \\
$\mathcal V^M$ & the set of mirror nodes \\
$\mathcal V$ & the set of nodes, $\mathcal V= \mathcal V^F \cup \mathcal V^M$\\
$\mathcal D$ & the set of commodities\\
$\mathcal P$ & the set of paths\\
$K$ & the number of link-disjoint paths for each pair of FSO nodes\\
$\kappa_e$ & the reliability of link $e \in \mathcal E$\\
$\delta^+ _{av}$ & =1 if arc $a \in \mathcal A$ is an outgoing arc of node $v \in \mathcal V$ and =0 otherwise.\\
$\delta^- _{av}$ & =1 if arc $a \in \mathcal A$ is an incoming arc of node $v \in \mathcal V$ and =0 otherwise.\\
$\Delta _{ev}$ & =1 if link $e \in \mathcal E$ is connected to node $v \in \mathcal V$ and =0 otherwise.\\
$\Lambda_{ep}$ & =1 if link $e \in \mathcal E$ belongs to path $p$ and =0 otherwise.\\
$s_d$ & the source  of commodity $d \in \mathcal D$\\
$t_d$ & the destination of commodity $d \in \mathcal D$\\
$p_e,q_e$ & the two end nodes of link $e \in \mathcal E$\\
$g_a,h_a$ & the head node and tail node of arc $a \in \mathcal A$\\
$l_a$ & the length of arc $a \in \mathcal A$\\
$L$ & the maximal transmission distance for an optical signal\\
$f$ & a mapping from arcs to links, for any arc $a \in \mathcal A$ corresponding to link $e \in \mathcal E$, we have $f(a)=e$\\
$\rho^+_{v}$ &  the set of outgoing ports for node $v \in \mathcal V$\\
$\rho^-_{v}$ & the set of incoming ports for node $v \in \mathcal V$\\
$\rho_{v}$ & the set of ports for node $v \in \mathcal V$, $\rho = \rho^+_{v} \cup \rho^-_{v}$\\
$\gamma_{av}$ & the index of the port connected to arc $a \in \mathcal A$ for node $v \in \mathcal V$\\
\hline
\end{tabular}
\hfill
\end{table}

\section{Problem description and formulation}\label{sec:formulation}
In this section we formulate the basic problem of this paper. Roughly speaking, the problem consists in constructing a $K$-connected subgraph while optimizing the number of the required FSO transceivers and mirrors for a given (potential) network graph with already installed fiber links. In the formulation we impose an additional requirement that the length of the mirror path is limited by a maximum transmission distance $L$.

The problem is formulated for a two-layer network model where the upper layer consists of the FSO links, fiber links and FSO nodes. The lower layer consists of the mirror links and all nodes, including the mirror nodes. This two-layer model is illustrated by an example in Fig. \ref{fig:layer}.

\begin{figure*}[!htbp]
\centering
\includegraphics[height=1.3in,width=5in]{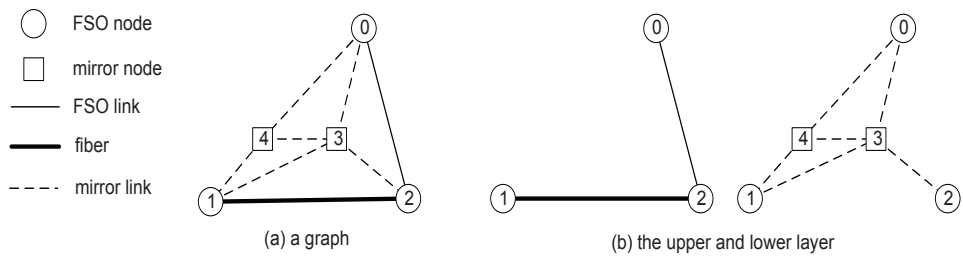}
\caption{An illustration of two layers for a graph}
\label{fig:layer}
\end{figure*}

Before proceeding to the problem formulation, we note that the strength of an optical signal is fixed when the signal is sent from an FSO transceiver (it is transmitted with a fixed power) and this strength is not increased when it is reflected by a mirror. Also, a signal incoming to a mirror cannot be split and forwarded to two different directions, it is just reflected to an outgoing direction. This is illustrated in Fig. \ref{fig:mirror} using two $2$-connected topologies. In Fig. \ref{fig:mirror}(a) nodes $4$ and $5$ have mirrors, which serve different commodity pairs, i.e., the commodity  $(0,3)$ and $(1,3)$. Note that nodes $2$ and $3$ could be connected by an FSO link (according to the Fact 1) if there is no the obstacle.
In Fig. \ref{fig:mirror}(b) path $(0,3,2,3,1)$ connects nodes $0$ and $1$. Path $(0,3,1)$ is not feasible since the total length of link $\{0,3\}$ and link $\{3,1\}$ is excessive. However, path $(0,3,2,3,1)$ is feasible since the signal will arrive firstly to the FSO node $2$ and then will be re-transmitted from with the maximum strength. In fact, this path consists of two feasible mirror paths $(0,3,2)$ and $(2,3,1)$. Further, for the mirror link $\{2,3\}$ two parallel beams should be established and two mirrors are deployed at node $3$ since the optical signal cannot be split by one mirror and forwarded to different directions. Note that case (a) requires 10 FSO transceivers and 2 mirrors while case (b) requires 8 FSO transceivers and 2 mirrors.

\begin{figure}[!hbtp]
  \centering
  \includegraphics[height=1.5in,width=4in]{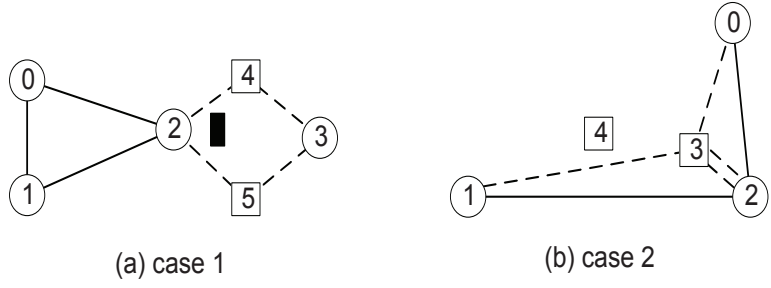}\\
  \caption{Optimal 2-connected topologies with the usage of mirrors}\label{fig:mirror}
\end{figure}

In the formulation each commodity is divided to $K$ flows and the corresponding $K$ link-disjoint paths so that each flow is transmitted over one path. To set up the formulation, we define the following variables.

\begin{tabular}{lp{16cm}}
$x_{a}^{kd}$ & binary variable indicating whether the $k$-th flow  of the commodity $d \in \mathcal D$ passes through the FSO arc $a \in \mathcal A^F$\\
$w_e$ & integer variable indicating the maximum number of flows of a commodity between the end nodes of the FSO link $e \in \mathcal E^F$\\
$u_e$ & binary variable indicating whether the direct beam connection is established for link $e \in \mathcal E$, $u_e=1$ (then the end nodes of $e$ must be in the line of sight), or not, $u_e=0$\\
$R_{e}^{kd}$ & binary variable indicating whether the link $e \in \mathcal E$ is used (or not) for the $k$-th flow of commodity $d$, $R_{e}^{kd}=1$ ($R_{e}^{kd}=0$)\\
$y_e$ &  integer variable expressing the number of mirrors used to connect the end nodes of a FSO link $e \in \mathcal E^F$\\
$r_{e}^{kd}$ &  binary variable indicating whether a mirror path corresponding to the FSO link $e \in \mathcal E^F$, for the $k$-th flow of commodity $d$, should be established, $r_{e}^{kd}=1$, or not, $r_{e}^{kd}=0$ \\
$z_{me}^{kd}$ & binary variable indicating whether the mirror arc $m \in \mathcal A^M$  is used by a mirror path realizing the FSO link $e \in \mathcal E^F$, for the $k$-th flow of the commodity $d$, $z_{me}^{kd}=1$, or not $z_{me}^{kd}=0$\\
$Z_{e'e}^{kd}$ & binary variable indicating whether the mirror link $e'\in \mathcal E^M$ is used by a mirror path realizing FSO link $e \in \mathcal E^F$, for the $k$-th flow of the commodity $d$, $Z_{e'e}^{kd}=1$,  or not, $Z_{e'e}^{kd}=0$ \\
$X_v$ &  binary variable indicating whether the mirror node $v \in \mathcal V^M$ is used for deploying mirrors, $X_v=1$, or not $X_v=0$.
\end{tabular}

The objective function that we minimize is given by
\begin{equation}\label{eq:objective}
c_1 \sum_{v \in V^F} \sum_{e \in  \mathcal E^F}  \Delta_{ev} w_e  + c_2 \sum_{v \in \mathcal V^M} X_v +  c_3  \sum_{e \in \mathcal E^F} y_e  - c_4 \sum_{e \in  \mathcal E \setminus  \mathcal E^O}  \kappa_e u_e.
\end{equation}

This is a multi-objective function and the positive coefficients $c_1,c_2,c_3,c_4$ are the weights of the four terms which are the sub-objectives. In \eqref{eq:objective}, the first term is the number of the FSO transceivers, where $w_e$ denotes the maximum number of flows of a commodity on the FSO link $e$. These flows should use link-disjoint paths, i.e., each flow is allocated either an FSO link $e$ or a mirror path. Since each flow starts from an FSO transceiver and ends at an FSO transceiver, the number of FSO transceivers needed for each flow on link $e$ is $\sum_{v \in V^F} \Delta_{ev} w_e$. Then the total number of  FSO transceivers,  that will be used in the optimal solution,  is $\sum_{v \in V^F} \sum_{e \in  \mathcal E^F}  \Delta_{ev} w_e $.
The  second term denotes the number of used mirror nodes.
The  third term denotes the number of mirrors which is the sum of all mirrors used in the mirror paths of the FSO links, $\sum_{e \in \mathcal E^F} y_e$.
The first four terms are related to the deployment cost. Since the cost of a FSO transceiver is higher than the leasing cost of building roof,  there is also the cost of mirrors, the relation of the costs is $c_1 \gg c_2 \gg c_3$.

The fourth term is the network reliability, which is the sum of the reliability over all links. The reliability of a link, denoted by $\kappa(e)$ and $0 \leq \kappa(e) \leq 1$. To maximize the network reliability, a  minus sign is placed before the term. Our objective is to balance the deployment cost and the system reliability.

To create a $K$-connected pair of FSO nodes, we transmit $K$ flows, with unitary volume for each path, from the source to the destination. This is formulated in \eqref{eq:constraints-0}, which is the flow conservation constraint.

\noindent \textbf{First-layer multicommodity flow constraints:}
\begin{equation}\label{eq:constraints-0}
\mbox{$\sum_{a \in \mathcal A^{F} \cup  \mathcal A^{O} } (\delta^+_{av}  x_{a}^{kd}- \delta^-_{av}  x_{a}^{kd})$}  =
                     \left\{ \begin{array}{l}
                        1, \;  v = s_d\\
                        -1, \;   v = t_d\\
                        0, \;  \mbox{otherwise.} \\
                     \end{array}
             \right.  k \in [1,K],v \in \mathcal V^F, d \in \mathcal D
\end{equation}

The $K$ flows of a commodity have to use $K$ link-disjoint paths. The link-disjoint constraints on  FSO links and  optical fibers are expressed in \eqref{eq:constraints-1}. If an FSO arc $a$ carries the $k$-th flow of commodity $d$, either an FSO link, indicated by $R^{kd}_{f(a)}$, is used or a mirror path, indicated by $r_{f(a)}^{kd}$, is established. At most one flow realizing a given commodity can pass through a given FSO link---this is assured by constraint \eqref{eq:constraints-1-2}. (Note that they can share a fiber link.) Further, constraints \eqref{eq:constraints-1-4} make sure that link $e$ is not established when it does not carry any flows using the direct beam
(this is important from the viewpoint of last term in the objective function). A link should also not be established if there is no line of sight, which is assured by \eqref{eq:constraints-1-5}.

\noindent \textbf{Link-disjoint constraint on FSO links and fiber links}
\begin{subequations}\label{eq:constraints-1}
\begin{align}
&\mbox{$ x_{a}^{kd} \le R^{kd}_{f(a)} + r_{f(a)}^{kd}$},  && a \in \mathcal A^F, d \in \mathcal D, k \in [1,K] \label{eq:constraints-1-1}\\
&\mbox{$ \sum_{k \in [1,K]} R^{kd}_{e} \le u_e$},  && e \in \mathcal E^F, d \in \mathcal D \label{eq:constraints-1-2}\\
&\mbox{$ \sum_{k \in [1,K]} \sum_{d \in \mathcal D} \sum_{a \in f^{-1}(e)} x_{a}^{kd} \ge u_e $} && e \in \mathcal E^F \label{eq:constraints-1-4}\\
&\mbox{$ \sum_{k \in [1,K]} \sum_{d \in \mathcal D} R_{e}^{kd} \ge u_e $} && e \in \mathcal E^F \label{eq:constraints-1-4}\\
& \mbox{$u_e = 0$},  &&  e \in \mathcal E^N. \label{eq:constraints-1-5}
\end{align}
\end{subequations}

Constraints \eqref{eq:constraints-1-5a} express the maximum number of flows belonging to a commodity between end nodes of FSO link $e$ ($w_e$) over all commodities. Since each flow needs a transmitting FSO transceiver and a receiving FSO transceiver, then the number of FSO transceivers needed in the optimal solution can be expressed as in the objective using $w_e$.

\noindent \textbf{Counting the maximum number of flows of a commodity between end nodes of a FSO link}
\begin{equation}\label{eq:constraints-1-5a}
\mbox{$\sum_{a \in f^{-1}(e)} \sum_{k \in [1,K]} x_{a}^{kd} \le w_{e}$}, \quad  e \in \mathcal E^F, d \in \mathcal D.
\end{equation}

In the second layer of the model, if $r_e^{kd} = 1$, then a mirror path that connects the end nodes of a FSO link $e$ for commodity $d$ must be found. This is addressed by the constraint \eqref{eq:constraints-2-1}.
Constraint \eqref{eq:constraints-2-2} assures that a mirror path cannot pass through any FSO nodes, except the origin and the destination. Constraint \eqref{eq:constraints-2-3} expresses that when a mirror arc $m$ is used in a mirror path for FSO link $e$, which carries the $k$-th flow of commodity $d$, then the corresponding link $f(m)$ is used for the $k$-th flow of commodity $d$.

\noindent \textbf{Second-layer multicommodity flow constraints:}
\begin{subequations}\label{eq:constraints-2}
\begin{align}
& \mbox{$\sum_{m \in  \mathcal A^{M}} (\delta^+_{mv} z_{me}^{kd}- \delta^-_{mv} z_{me}^{kd})$} =
                     \left\{ \begin{array}{l}
                        r_{e}^{kd}, \;  v = p_e\\
                        -r_{e}^{kd}, \;   v = q_e\\
                         0, \;  v \in \mathcal V^M \\
                     \end{array}
             \right. \quad k \in [1,K],  e \in \mathcal E^{F}, d \in \mathcal D \label{eq:constraints-2-1}\\
& \mbox{$\sum_{m \in  \mathcal A^{M}} (\delta^+_{mv} + \delta^-_{mv})  z_{me}^{kd}$} =0, \quad  k \in [1,K], \quad  v \in \mathcal V^F \setminus \{p_e,q_e\}, \; e \in \mathcal E^{F}, d \in \mathcal D \label{eq:constraints-2-2}\\
& \mbox{$z_{me}^{kd}$} \leq Z_{f(m)}^{kd}, \quad  k \in [1,K],  m \in  \mathcal A^{M}, \; e \in \mathcal E^{F}, d \in \mathcal D. \label{eq:constraints-2-3}
\end{align}
\end{subequations}

At most one of the flows realizing a given commodity can use a given mirror link---this is assured by \eqref{eq:constraints-3-1} since $u_e$ is binary. This constraint assures also that a mirror link $e$ is established when there are some flows on it. If there are no flows using a mirror link $e$ then it is not---this is assured by \eqref{eq:constraints-3-2}.

\noindent \textbf{Disjoint mirror links constraints:}
\begin{subequations}\label{eq:constraints-3}
\begin{align}
\mbox{$\sum_{k=1}^K Z_{e}^{kd} \leq u_{e}$}, \quad d \in \mathcal D, e \in \mathcal E^M \label{eq:constraints-3-1}\\
\mbox{$\sum_{d \in \mathcal D} \sum_{k=1}^K  \sum_{e' \in \mathcal E} \sum_{m \in f^{-1}(e)}z_{me'}^{kd} \geq u_{e}$}, \quad e \in \mathcal E^M. \label{eq:constraints-3-2}
\end{align}
\end{subequations}

The inequality \eqref{eq:constraints-4} expresses the number of mirrors used to connect end nodes of FSO link $e$. In the left hand side, the term $\sum_{m \in \mathcal A^M} z_{me}^{kd} - r_{e}^{kd}$ denotes the number of mirrors used in a mirror path for FSO link $e$ deployed by the $k$-th flow of a commodity $d$. Then, for a commodity $d$, the number of mirrors needed in order to connect the end nodes of a FSO link $e$ is the summation over the number of mirrors for all the flows. Then, that number is the maximal number of mirrors needed for all the commodities, as described in Proposition \ref{pp:mirrorNum}.

\noindent \textbf{Counting the number of mirrors:}
\begin{equation}\label{eq:constraints-4}
\mbox{$\sum_{k=1}^K (\sum_{m \in \mathcal A^M} z_{me}^{kd} - r_{e}^{kd}) \leq y_{e}$}, \; e \in \mathcal E^{F}, d \in \mathcal D. \\
\end{equation}

\begin{proposition}\label{pp:mirrorNum}
The number of deployed mirrors between the end nodes of a FSO link is the maximum number of mirrors needed for all the commodities.
\end{proposition}
\begin{proof}
If a flow passes through the end nodes of a FSO link $e$ and there is a need for
a mirror path, $r_e^{kd}=1$, then the number of used mirrors is $\sum_{m \in \mathcal A^M} z_{me}^{kd} - r_{e}^{kd}$. Note that if $r_e^{kd}=0$, then $z_{me}^{kd}=0$ which is assured by \eqref{eq:constraints-2-1}.
The number of used mirrors for the commodity on the FSO link $e$ is $\sum_{i=1}^K (\sum_{m \in \mathcal A^M} z_{me}^{kd} - r_{e}^{kd})$, since two or more flows cannot pass through the same mirror link.
Note that a mirror path for a FSO link can be used for different commodities. Thus the number of used mirrors between the end nodes of a FSO link is the maximum number of mirrors needed over all the commodities.
\end{proof}

The length of each mirror path should not be bigger than the maximum transmission distance, as the optical signal attenuates during its transmission along a mirror path. This is assured by inequality \eqref{eq:constraints-5}. When the optical signal arrives at an FSO node, which is not the destination, it will be transmitted with the fixed transmission power.

\noindent \textbf{Limiting the length of a mirror path:}
\begin{equation}\label{eq:constraints-5}
\mbox{$\sum_{m \in \mathcal A^M} l_m z_{me}^{kd} \leq L$}, \quad k \in [1,K], e \in \mathcal E^{F}, d \in \mathcal D.
\end{equation}

If a mirror link is established, the corresponding mirror node should be leased, which is assured by \eqref{eq:constraints-6}.

\noindent \textbf{Mirror node constraint:}
\begin{equation}\label{eq:constraints-6}
u_e \leq X_v, \quad v \in \mathcal V^M, e \in \mathcal E^M, \Delta_{ev} = 1.
\end{equation}

In summary, the mathematical formulation for cost-efficient and resilient backhaul network design utilizing FSO transceivers, mirrors and optical fibers (CRBND) is given below.
\begin{align}
&& \mbox{\textbf{CRBND}} &\nonumber \\
&& \mbox{Minimize} \quad  & \eqref{eq:objective}  \nonumber\\
&& \mbox{s.t.} \quad & \mbox{First-layer multicommodity flow constraints: } \eqref{eq:constraints-0}  \nonumber\\
&& & \mbox{Link-disjoint constraint on FSO links and fiber links:} \eqref{eq:constraints-1}  \nonumber\\
&& & \mbox{Counting the maximum number of flows of a commodity}  \nonumber\\
&& & \mbox{between end nodes of a FSO link:} \eqref{eq:constraints-1-5a}  \nonumber\\
&& & \mbox{Second-layer multicommodity flow constraints:} \eqref{eq:constraints-2}  \nonumber\\
&& & \mbox{Disjoint mirror links constraints: } \eqref{eq:constraints-3}  \nonumber\\
&& & \mbox{Counting the number of mirrors: } \eqref{eq:constraints-4}  \nonumber\\
&& & \mbox{Limiting the length of a mirror path: } \eqref{eq:constraints-5}  \nonumber\\
&& & \mbox{Mirror node constraint: } \eqref{eq:constraints-6}  \nonumber
\end{align}

This network design problem CRBND is $NP$-hard since the special case of this problem with no mirror nodes is an instance of the problem of finding the minimal $K$-edge-connected graph which is known to be $NP$-hard \cite{Garey1980}.

\section{An efficient sequential computation approach}\label{sec:approach}
The time required to solve the CRBND increases dramatically as the size of the network increases. Therefore, for large networks, we need to find an efficient and practical method. In this section, we propose a sequential computation approach, i.e., finding $K$ link-disjoint paths in a manner of commodity-by-commodity. The nodes and the links that are used in the optimal solution of  a commodity will not be counted in the objective function for calculating all subsequent commodities.

The steps for this approach are given below.

\textbf{Step 1}. The commodities are sorted in the descending order based on the distance of the source and the destination. Thus, the commodity with the longest distance between the source and the destination is disposed first. The reason is that the commodity with the longer distance is inclined to uses more FSO transceivers and mirrors, which will be reused by subsequent commodities to save the total cost.

\textbf{Step 2}. For each commodity, we need to establish $K$ link-disjoint paths. We develop a new non-compact formulation, called the link-path formulation. This formulation is based on the notation used for paths, which is shown in \eqref{model:oneDemand}. We denote the set of available paths for the considered commodity as $\mathcal P$, define the constant $\Lambda_{ep}$ to represent whether link $e \in \mathcal E$ belongs to the path $p \in \mathcal P$, $\Lambda_{ep}=1$, or not $\Lambda_{ep}=0$, and define the constant $\Theta_{ap}$ to represent whether arc $a \in \mathcal A$ belongs to the path $p \in \mathcal P$, $\Theta_{ap}=1$, or not $\Theta_{ap}=0$. $\Lambda_{ep}$ and $\Theta_{ap}$ can be computed by solving the pricing problem which will be introduced later.

The introduced variables are defined below.

\begin{tabular}{lp{16cm}}
$X_p$ &  binary variable, $X_p=1$ if path $p \in \mathcal P$ is used; $X_p=0$ otherwise\\
$U_e$ &  binary variable, $U_e=1$ if link $e \in \mathcal E$ is established; $U_e=0$ otherwise\\
$Y_v$ &  binary variable, $Y_v=1$ if mirror node $v \in \mathcal V^M$ is used; $Y_v=0$ otherwise\\
\end{tabular}

The link-path formulation is given as below.
\begin{subequations}\label{model:oneDemand}
\begin{align}
\mbox{Min} \;\; &  \mbox{$ c_1 \sum_{v \in \mathcal V^F}  \sum_{a \in  \mathcal A \setminus \mathcal A^O} \sum_{p \in \mathcal P}  \delta_{av} \Theta_{ap} X_{p}  + c_2 \sum_{v \in \mathcal V^M} Y_v $} \nonumber \\
&  \mbox{$ + \frac{c_3}{2} \sum_{a \in \mathcal A^M}\sum_{v \in \mathcal V^M} \sum_{p \in \mathcal P}  \delta_{av} \Theta_{ap}  X_p  - c_4 \sum_{e \in  \mathcal E \setminus \mathcal E^O} \kappa_{e} U_e $} \label{eq:oneDemand-1}\\
& \mbox{$\sum_{p \in \mathcal P} X_{p} =K$} \label{eq:oneDemand-2}\\
& \mbox{$\sum_{p \in \mathcal P} \Lambda_{ep} X_{p} = U_{e}$}, \; e \in \mathcal E^F \cup \mathcal E^M  \label{eq:oneDemand-3}\\
& \Delta_{ev} U_e \leq Y_v, \; v \in \mathcal V^M, e \in \mathcal E^M  \label{eq:oneDemand-5}
\end{align}
\end{subequations}

The four terms in the objective function \eqref{eq:oneDemand-1} denote the number of the FSO transceivers,  the number of used mirror nodes, the number of mirrors and the network reliability. In this objective, the first term and the third term are different from their counterparts in \eqref{eq:objective}, however the equivalence is proved in Proposition \ref{pp:oneDemand}.
The constraints \eqref{eq:oneDemand-2} assure that $K$ paths are selected and they are link-disjoint, i.e., two paths cannot share an FSO link or a mirror link, which is assured by constraint \eqref{eq:oneDemand-3}. Note that an optical fiber can be shared by different paths.
Constraints \eqref{eq:oneDemand-3} also assure that an FSO link or a mirror link is deployed if and only if there is a flow on it. Finally, constraints \eqref{eq:oneDemand-5} force that a mirror node should be leased when there is a deployed link attached to it.

\begin{proposition}\label{pp:oneDemand}
In the model \eqref{model:oneDemand}, the number of FSO transceivers is $\sum_{v \in \mathcal V^F}  \sum_{a \in  \mathcal A \setminus \mathcal A^O} \sum_{p \in \mathcal P}  \delta_{av} \Theta_{ap} X_{p} $ and the number of mirrors is $\frac{1}{2} \sum_{a \in \mathcal A^M}\sum_{v \in \mathcal V^M} \sum_{p \in \mathcal P}  \delta_{av} \Theta_{ap} X_p$. 
\end{proposition}
\begin{proof}
A path in the optimal solution of \eqref{model:oneDemand} starts from an FSO node and terminates at an FSO node. The intermediate node can be a mirror node or an FSO node. Note that two oppositely directed arcs corresponding to a same link may appear in a path, e.g., path $(0,3,2,3,1)$ in Fig. \ref{fig:mirror}(b). Thus, the total number of FSO transceivers for path $p$ can be counted in this way: for each FSO node, count the number of attached arcs included in path $p$, and then sum up over all FSO nodes, which is expressed  as $ \sum_{v \in \mathcal V^F}  \sum_{a \in  \mathcal A \setminus \mathcal A^O} \delta_{av} \Theta_{ap} X_{p} $. The number of mirrors in path $p$ is counted in a similar way: for each mirror node, count the number of attached mirror arcs included in path $p$, sum up over all mirror nodes and divide by 2 (a mirror corresponds two mirror links), which is expressed as $\sum_{v \in \mathcal V^M}  \sum_{a \in \mathcal A^M} \delta_{av} \Theta_{ap} X_p$. The paths for one commodity are link-disjoint, meaning that each arc $a \in \mathcal A$ will not be used in different paths. Therefore, the total number of FSO transceivers (mirrors) can be obtained by summing up over all paths.
\end{proof}

The optimization problem in \eqref{model:oneDemand} cannot be solved directly, since we cannot pre-define a set of paths where the optimal paths are included. A practical method is to solve its linear relaxation by path generation, i.e., iteratively generating paths until the optimum is achieved \cite{pioro2004}.

Regarding the linear relaxation of \eqref{model:oneDemand}, which is called master problem, let $\lambda$, $\pi_{e}, e \in \mathcal E$ denote the dual variables for constraints \eqref{eq:oneDemand-2} and \eqref{eq:oneDemand-3} respectively. Given an initial set of paths $\mathcal P^*$, we solve the master problem and obtain the optimal dual variables $\lambda^*, \pi_{e}^*, e \in \mathcal E$.
In the given graph, if we can find a new path $p$, such that $\sum_{e \in \mathcal E} \Lambda_{ep} \pi^*_{e} \leq \lambda^*$, i.e., the path has shorter length than $\lambda^*$, with $\pi^*_{e}, e \in \mathcal E$ denoting the weights of the arcs. Then, this path is added to the current set of paths and the master problem is solved again to obtain the new optimal dual values. This procedure is repeated until no paths can be generated, and it is called path generation, see \cite{pioro2004}.

The problem of finding a path, is a pricing problem, in our case is equivalent to the generation of a shortest path with respect to $\pi^*_{e}, e\in \mathcal E$, which satisfies the length constraint for any inside mirror path. For example, consider the path $(0,3,2,3,1)$ in Fig. \ref{fig:mirror}(b), the length of the mirror path $(0,3,2)$ and the mirror path $(2,3,1)$ has to be less than $L$.

To model the pricing problem, we define a set of ports for every node. Each port is connected to an arc attached to the node. The set of index of the ports for node $v$ is represented by $\rho_v$, the set of outgoing ports is represented by $\rho^+_v$ and the set of incoming ports is represented by $\rho^-_v$. Let $\gamma_{av}$, where $ a \in \mathcal A$ and $v \in \mathcal V$, be the index of the port in node $v$ connecting arc $a$. The pricing problem is formulated in \eqref{model:pricing}, while the variables are defined as follows.

\begin{tabular}{lp{16cm}}
$x_a$ & binary variable, $x_a=1$ if arc $a$ is used; $x_a=0$ otherwise\\
$u_e$ & binary variable, $u_e=1$ if link $e$ is used; $u_e=0$ otherwise\\
$y^+_{vi}$ & continuous variable, representing the transmitted distance when the signal is sent out from outgoing port $i \in \rho^+_v$ of node $v \in \mathcal V$\\
$y^-_{vi}$ & continuous variable, representing the transmitted distance when the signal arrives at incoming port $i \in \rho^-_v$ of node $v \in \mathcal V$\\
$z_{vij}$ & binary variable, $z_{vij}=1$ if the signal in incoming port $j \in \mathcal \rho^-_{v}$ is forwarded to outgoing port $i \in \mathcal \rho^+_{v}$ for a node $v \in \mathcal V$.\\
\end{tabular}

The objective \eqref{eq:prcing-1}, is to minimize the total weighted links, i.e., finding a shortest path. Constraints \eqref{eq:prcing-2} are the flow conservation rule, which is used to find a path that connects the source $s$ and the destination $t$ of a commodity.

When an optical signal is transmitted by an FSO node, the transmitted distance is set to be $0$, this is stated by the constraint \eqref{eq:prcing-3}.
When an optical signal is transmitted along an arc $a$, from the outgoing port $\gamma_{a g_a}$ of node $g_a$ to the incoming port $\gamma_{a h_a}$ of node $h_a$, the length of the arc, $l_a$, is accumulated. This is expressed by constraints \eqref{eq:prcing-4}.

When an optical signal arrives at a mirror node $v \in \mathcal V^M$, it will be forwarded from an incoming port to an outgoing port without changing the transmitted distance. This is assured by constraints \eqref{eq:prcing-5} and \eqref{eq:prcing-6}. These constraints are equivalent to $|y^+_{vi} - y^-_{vj}| \leq L (1-z_{vij})$. If $z_{vij} =1$, i.e., the optical signal arrives at the incoming port $j$  and is forwarded to the outgoing port $i$ of node $v$, then we have $y^+_{vi} = y^-_{vj}$. Otherwise, i.e., $z_{vij} =0$, the inequalities hold trivially.
Note that the maximum transmitted distance for an optical signal arriving at a node should be smaller than $L$, which is assured by constraints \eqref{eq:prcing-12}. In this way, the constraint for each mirror path is indicated.

Further, constraints \eqref{eq:prcing-7} ensure that, for a mirror node $v$, if an outgoing arc $a$ is not used, any signal from incoming port $i$ should not be forwarded to outgoing port $\gamma_{ag_a}$ (the port of node $g_q$ connecting arc $a$). Similarly, constraints \eqref{eq:prcing-8} guarantee that for a mirror node $v$, if an incoming arc $a$ is not used, then the signal in the corresponding incoming port $\gamma_{ah_a}$ should not be forwarded to any outgoing port $i$. Constraints \eqref{eq:prcing-9} and \eqref{eq:prcing-10} guarantee that signals in an incoming port can only be forwarded to one  outgoing port and vice visa.  Finally, constraints \eqref{eq:prcing-11} assure that a link is used when any corresponding arc is used.

After solving the price problem, we can obtain $\Lambda_{ep}$ and $\Theta_{ap}$ for current path $p$, based on the optimal value of $x_a$ and $u_e$ respectively.
\begin{subequations}\label{model:pricing}
\begin{align}
&\mbox{ minimize  $\sum_{e \in \mathcal E} \pi^*_e u_{e}$} \label{eq:prcing-1}\\
& \mbox{$\sum_{a \in \mathcal A}\delta^{+}_{av} x_{a} - \sum_{a \in \mathcal A } \delta^{-}_{av} x_{a}$} = \left \{
\begin{aligned}
& 1, \;  v=s\\
& -1, \; v=t\\
& 0, \; \mbox{otherwise}
\end{aligned}
\right.  \label{eq:prcing-2}\\
& y^+_{vi} = 0,  \; i \in \rho^+_{v}, v \in \mathcal V^F  \label{eq:prcing-3}\\
& \mbox{$y^-_{h_ai} =  y^+_{g_aj} + l_a x_a$}, \; i = \gamma_{a h_a}, j = \gamma_{a g_a}, a \in \mathcal A  \label{eq:prcing-4}\\
& \mbox{$y^+_{vi} - y^-_{vj} \leq L (1-z_{vij})$}, \; i \in \rho^+_{v}, j \in \rho^-_{v},v \in \mathcal V^M  \label{eq:prcing-5}\\
& \mbox{$y^+_{vi} - y^-_{vj} \geq L (z_{vij}-1)$}, \; i \in \rho^+_{v}, j \in \rho^-_{v},v \in \mathcal V^M  \label{eq:prcing-6}\\
& z_{v\gamma_{av}i} \leq x_a, v = g_a, \; a \in \mathcal A, i \in \rho^-_{v} \label{eq:prcing-7}\\
& z_{vi\gamma_{av}} \leq x_a, v = h_a, \; a \in \mathcal A, i \in \rho^+_{v} \label{eq:prcing-8}\\
& \mbox{$\sum_{i \in \rho^+_{v}} z_{vij} = 1$}, \; j \in \rho^-_{v},v \in \mathcal V^M  \label{eq:prcing-9}\\
& \mbox{$\sum_{j \in \rho^-_{v}} z_{vij} = 1$}, \; i \in \rho^+_{v},v \in \mathcal V^M  \label{eq:prcing-10}\\
& x_{a} \leq u_{f(a)}, \; a \in \mathcal A \label{eq:prcing-11}\\
& 0 \leq y^+_{vi} \leq L, \; 0 \leq y^-_{vj} \leq L, \; i \in \rho^+_{v}, j \in \rho^-_{v}, v \in \mathcal V  \label{eq:prcing-12}
\end{align}
\end{subequations}

In order to solve the master problem, an initial feasible set of paths must be provided. We introduce the auxiliary variable $Z$ and then solving the \eqref{model:initialPaths} by column generation. During the path generation, we can stop the procedure before the optimum is reached when $Z=0$. If in the optimum $Z \neq 0$,  then, there are not $K$ link-disjoint paths for the considered commodity. 
\begin{subequations}\label{model:initialPaths}
\begin{align}
\mbox{Min} \;\; & Z \label{eq:initialPaths-1}\\
& \mbox{$\sum_{p \in \mathcal P} X_{p} =K + Z$} \label{eq:initialPaths-2}\\
& \mbox{$\sum_{p \in \mathcal P} \Lambda_{ep} X_{p} \leq 1$}, \; e \in \mathcal E^F \cup \mathcal E^M  \label{eq:initialPaths-3}\\
& X,Y,Z \mbox{-- continuous} \nonumber
\end{align}
\end{subequations}

\textbf{Step 3}. When the linear relaxation of the model \eqref{model:oneDemand} is solved, a feasibility check should be carried out to assure that there are $K$ link-disjoint paths. A easy way is solving the equivalent binary problem of \eqref{model:initialPaths}. If the model is infeasible, then we have to solve the CRBND for the current commodity in order to obtain feasible paths.

After computing paths for current commodity, we obtain the deployed links $\{e \in \mathcal A: U_e =1\}$, mirrors nodes $\{v \in \mathcal V^M: Y_v =1 \}$ and the set of paths $\mathcal P_d$. For calculating paths of subsequent commodities, the deployed links and nodes will not be included in the cost.

To accelerate the path generation for subsequent commodities, we break down the obtained paths of the current commodity to compute some paths which can be used for subsequent commodities.
For current commodity $d$, let $\xi^p_{d}$ be a set of nodes corresponding to path $p \in \mathcal P_d$, $\xi^p_{dn}$ will be the $n$-th node in the path $p$, and $\xi_d =\{\xi^p_{d}, p \in \mathcal P_d\}$

The procedure of breaking down the paths is shown by the Algorithm \ref{alg:breakDown}.

\begin{algorithm}
\caption{Breaking down paths for commodity $d$}
\label{alg:breakDown}
\begin{algorithmic}[1]
\REQUIRE $\mathcal P_d$
	\FORALL{$p \in \mathcal P_d$}
        \FOR{$m = 1,2,.. , |\xi^p_d|-1$}
            \FOR{$n = m+1,.. , |\xi^p_d|$}
                \IF{$\xi^p_{dm},\xi^p_{dn} \in \mathcal V^F$}
                    \STATE $d'$ is the commodity corresponding to $(\xi^p_{dm},\xi^p_{dn})$
                    \IF {$\{\xi^p_{dm},\xi^p_{dm+1},...,  \xi^p_{dn}\} \not\in \mathcal \xi_{d'}$}
                        \STATE  $\xi^{|\mathcal P_{d'}|+1}_{d'} = \{\xi^p_{dm},\xi^p_{dm+1},...,  \xi^p_{dn}\}$
                        \STATE $\mathcal P_{d'} \leftarrow |\mathcal P_{d'}| +1$
                    \ENDIF
                \ENDIF
            \ENDFOR
        \ENDFOR
    \ENDFOR
\end{algorithmic}
\end{algorithm}

\textbf{Step 4}. When we obtain paths for all commodities, we consider all paths jointly for optimizing the deployment cost and the network reliability. Using all obtained paths, we solve a mixed integer programming model, which is obtained from model \eqref{model:oneDemand} by using $X_{dp}$ instead of replacing $X_{p}$. $X_{dp}$ are binary variables indicating if the path $p$ is deployed by commodity $d$. The constant $\Lambda_{ep}$ is also replaced by $\Lambda_{edp}$, representing whether link $e$ is included in path $p$ of commodity $d$.
Constraints \eqref{eq:oneDemand-3} are changed to $\mbox{$\sum_{p \in \mathcal P} \Lambda_{edp} X_{dp} \leq U_{e}, \; e \in \mathcal E^F \cup \mathcal E^M, d \in \mathcal D$}$ and adding constraints $\mbox{$\sum_{d \in \mathcal D} \sum_{p \in \mathcal P} \Lambda_{edp} X_{dp} \geq U_{e}, \; e \in \mathcal E^F \cup \mathcal E^M$}$ to ensure that link $e$ will not be established when there are no flows on it.

In summary, the sequential computation approach is presented in Algorithm \ref{alg:sequential}, where $FL$ is the set of deployed links, $FM$ is the set of used mirror nodes, and $\eta$ represents a newly generated path.

\begin{algorithm}
\caption{The sequential computation approach}
\label{alg:sequential}
\begin{algorithmic}[1]
\REQUIRE $\mathcal G(\mathcal V, \mathcal E)$
\STATE Sort $d \in \mathcal D$ in descending order with respect to distance of the source and the destination.
\STATE $\mathcal P_d = \{\}, \xi_d =\{\}, d \in \mathcal D$, $FL=\{\}, FM=\{\}$
	\FORALL{$d \in \mathcal D$}
       \WHILE {True}
         \STATE $\lambda^*, \pi_e^*, Z^* =$ solve the model \eqref{model:initialPaths} ($\mathcal P_d$)
         \IF {model \eqref{model:initialPaths} is infeasible}
             \STATE return \{\}; There are no feasible $K$ link-disjoint paths for current commodity.
         \ENDIF
         \STATE $\eta=$solve the model \eqref{model:pricing} $(\pi_e^*, e \in \mathcal E)$
         \STATE $p'=|\mathcal P_d|+1$ and compute $\Delta_{ep'}, e \in \mathcal E$ and $\Theta_{ap'}, a \in \mathcal A$ according to $\eta$
         \IF {$\sum_{e \in \eta} \pi^*_e < \lambda^*$ and $Z^*>0$}
             \STATE $\mathcal P_d \leftarrow p'$
             \STATE $\xi_d \leftarrow \eta$
         \ELSE
             \STATE break
         \ENDIF
       \ENDWHILE
       \WHILE {True}
         \STATE $U_e^*, Y_v^*,\lambda^*, \pi_e^*=$ solve the model \eqref{model:oneDemand} $(P_d,FL,FM)$
         \STATE Repeat 9-12, but remove $Z^*>0$ in 11.
       \ENDWHILE
       \STATE Feasibility check: $Z^*=$ solve model \eqref{model:initialPaths} ($\mathcal P^*_d$) with variables set to be binary
       \IF { $Z^* >0$}
       \STATE Solve the model CRBND for $d$ to obtain new feasible $\mathcal P_d$
       \IF { the model CRBND is still infeasible}
       \STATE return \{\}; There are no feasible $K$ link-disjoint paths for current commodity.
       \ENDIF
       \ENDIF
       \STATE $FL \leftarrow \{e \in \mathcal E: U_e^* =1 \}$, $FM \leftarrow \{v \in \mathcal V^M: Y_v^* =1 \}$
       \STATE Breaking down $\mathcal P_d$ by Algorithm 1
    \ENDFOR
    \RETURN $\mathcal (P_d^*, d \in \mathcal D$) = solve model \eqref{model:oneDemand} with variables $X_{dp}$($\mathcal P_d, d \in \mathcal D$)
\end{algorithmic}
\end{algorithm}

The Algorithm \ref{alg:sequential} computes $K$ link-disjoint routing paths for each commodity. However, the number of FSO transceivers and mirrors is not explicitly given. It is not correct to count them on each path and then sum them all together. Since some FSO transceivers and mirrors used for one commodity can be reused for other commodities.
Thus, we propose the Algorithm \ref{alg:deploy} to count FSO transceivers and mirrors based on obtained paths. Let $F$ and $M$ denote the number of FSO transceivers and the number of mirrors respectively.

For each path of a commodity, following the order of links in the path, we take a mirror path, count FSO transceivers and mirrors and then select the next one. Note that each mirror path has its own FSO transceivers and mirrors.
To avoid duplicated counting FSO transceivers and mirrors for different commodities, each mirror path is counted only  once.

\begin{algorithm}
\caption{Counting FSO transceivers and Mirrors}
\label{alg:deploy}
\begin{algorithmic}[1]
\REQUIRE $\mathcal P_d, \xi_d, d \in \mathcal D$
    \STATE F=0,M=0
    \FORALL{$d \in \mathcal D$}
	  \FORALL{$p \in \mathcal P_d$}
          \FOR{$m = 1,2,.. , |\xi^p_d|-1$}
              \FOR{$n = m+1,.. , |\xi^p_d|$}
                  \IF{$\xi^p_{dm},\xi^p_{dn} \in \mathcal V^F$}
                      \IF{the  mirror path $\{\xi^p_{dm},\xi^p_{d(m+1)}, ...\xi^p_{dn}\}$ has not been disposed}
                          \STATE F=F+1,M=M+n-m
                          \STATE m = n
                      \ENDIF
                  \ENDIF
              \ENDFOR
          \ENDFOR
      \ENDFOR
     \ENDFOR
\RETURN F,M
\end{algorithmic}
\end{algorithm}

\section{Numerical Results} \label{sec:numerical}
In this section we provide numerical results regarding optimal topologies and optimal deployment of network devices under certain realistic scenarios. The impact of different parameters of the design of the network is analyzed, furthermore, the efficiency of the proposed algorithm is compared with the exact model.

The model and the proposed algorithms are evaluated through experiments based on both realistic and randomly generated data. We consider the following values for the parameters in all of the experiments. We set $\Gamma_{th} = 0.88$, $I_{th}/I_0 = 0.8$,  $\lambda = 1550$ nm which is a commonly used wavelength in optical communications, and $C^2_n = 10^{-15}$ $m^{-2/3}$ expressing a moderate atmospheric turbulence.

An FSO link or a mirror link can be established if and only if $\Gamma(l_e) \geq \Gamma_{th}$. Then, the maximum transmission distance, $L=1400$ m, is obtained. The objective function can be obtained in two steps. The first step is the cost including he number of FSO transceivers, the number of leased mirror nodes and the number of mirrors. We set the weights as $c_1 : c_2 :c_3 = 4:2:1$. The second step is the total reliability of FSO links and mirror links, the coefficient for the reliability is $c_4$. The trade-off between the two type of objectives will be analyzed in the remaining of this section.

The Gurobi solver \cite{Gurobi} is used for the integer and linear programming problems. All the computations were executed on a Windows XP computer equipped with a dual core Intel $2.53$ GHZ CPU and $1.93$GB RAM.

The realistic data is a snapshot of the district around Alexanderplatz of Berlin, and includes $11$ base stations in an area of $3000 \times 3000 m^2$. The realistic data are provided by the EU MOMENTUM project \cite{MOMENTUM}. Fig. \ref{fig:scenarios} is a map where the red circles are base stations, the blue squares are the candidate locations, building roofs, for deploying mirrors. The black unnumbered squares represent high buildings which are treated as obstacles.

Regarding the randomly generated data, we generate the same number of FSO nodes and mirror nodes in a square area of $3000 \times 3000 m^2$ randomly. The set of links consists of any pair of nodes whose distance is smaller than $L$.

\begin{figure}
        \centering
        \begin{subfigure}[b]{0.25\textwidth}
                \centering
                \includegraphics[trim = 1.2in 0.3in 1.2in 0.3in, clip, width=\textwidth]{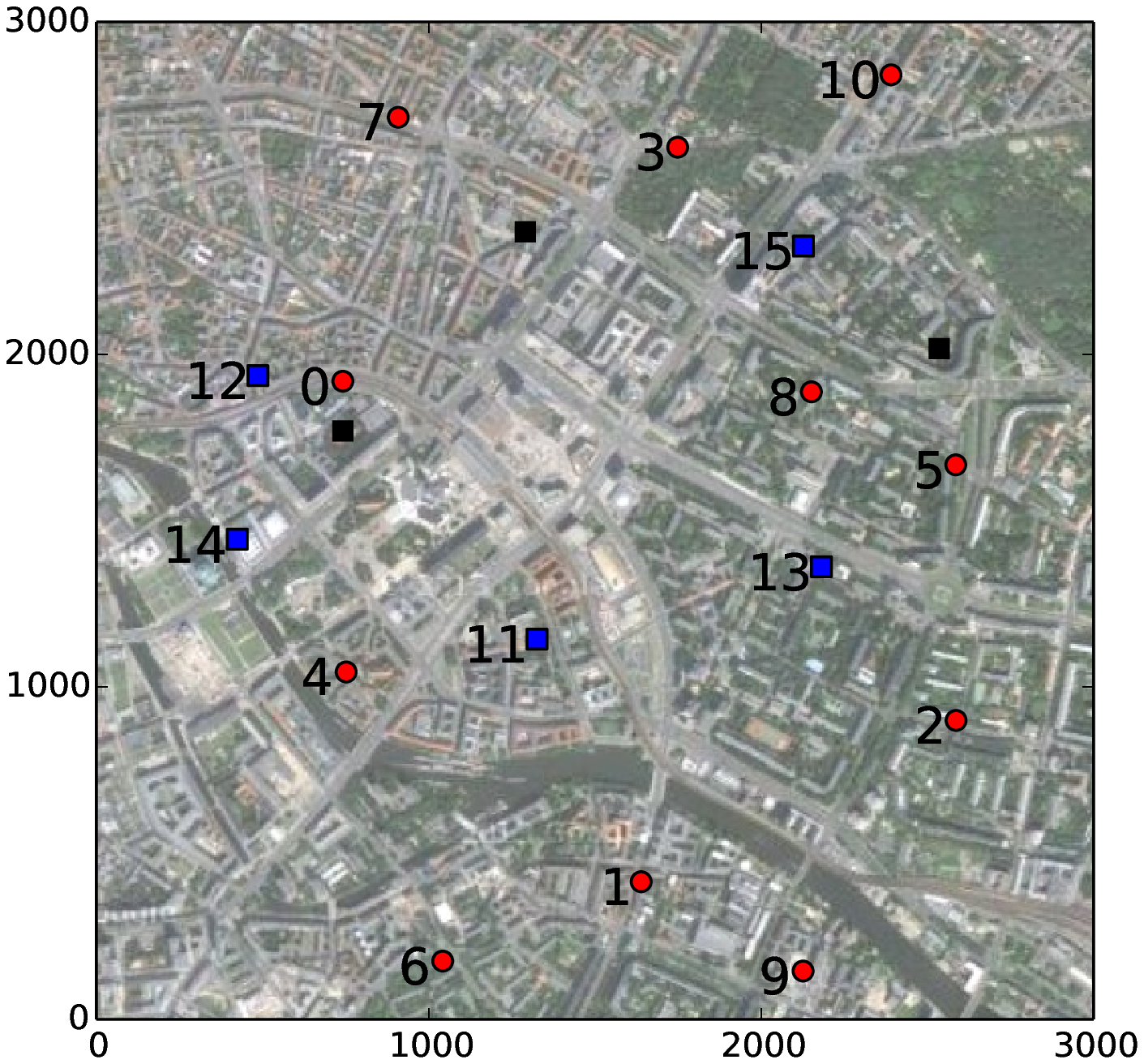}
                \caption{R1}
                \label{fig:berlin1}
        \end{subfigure}
        \begin{subfigure}[b]{0.25\textwidth}
                \centering
                \includegraphics[trim = 1.2in 0.3in 1.2in 0.3in, clip, width=\textwidth]{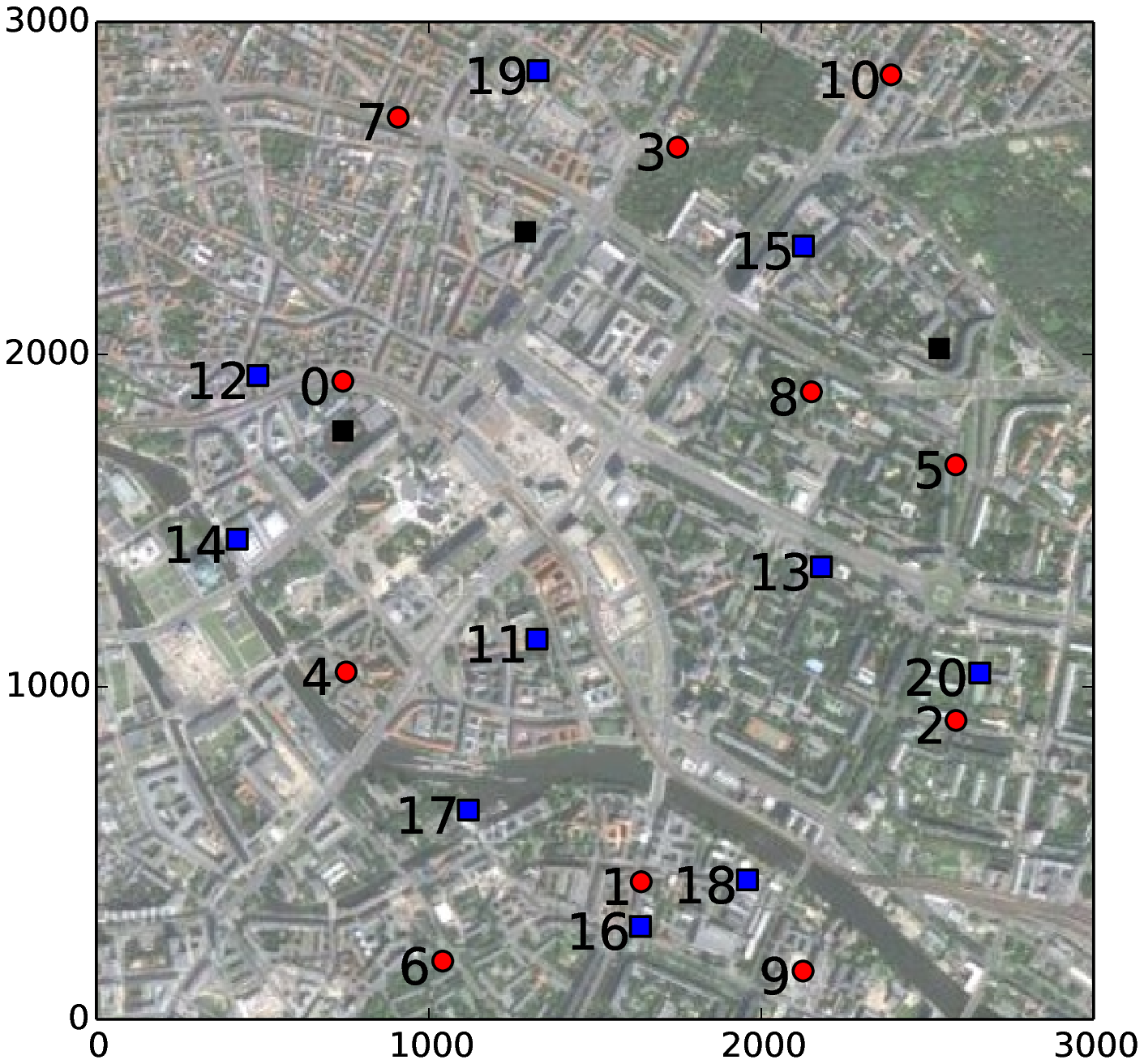}
                \caption{R2}
                \label{fig:berlin2}
        \end{subfigure}
        \caption{Two scenarios of base stations in Berlin}\label{fig:scenarios}
\end{figure}
\begin{figure*}[!hbtp]
        \centering
        \begin{subfigure}[b]{0.3\textwidth}
                \centering
                \includegraphics[trim = 1.2in 0.3in 1.2in 0.3in, clip, width=2in]{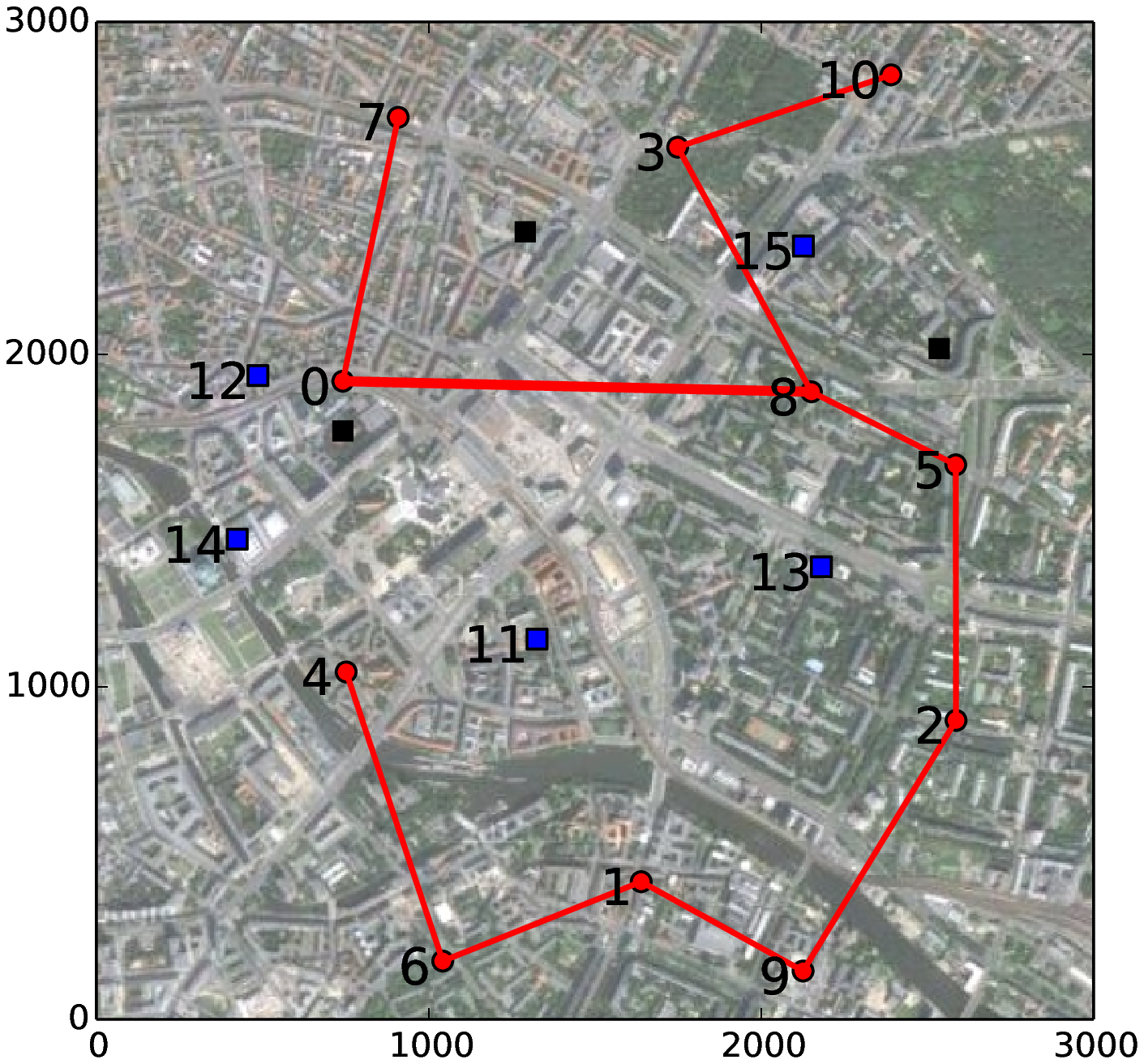}
                \caption{1-connected}
                \label{fig:1-connected}
        \end{subfigure}
        \begin{subfigure}[b]{0.3\textwidth}
                \centering
                \includegraphics[trim = 1.2in 0.3in 1.2in 0.3in, clip, width=2in]{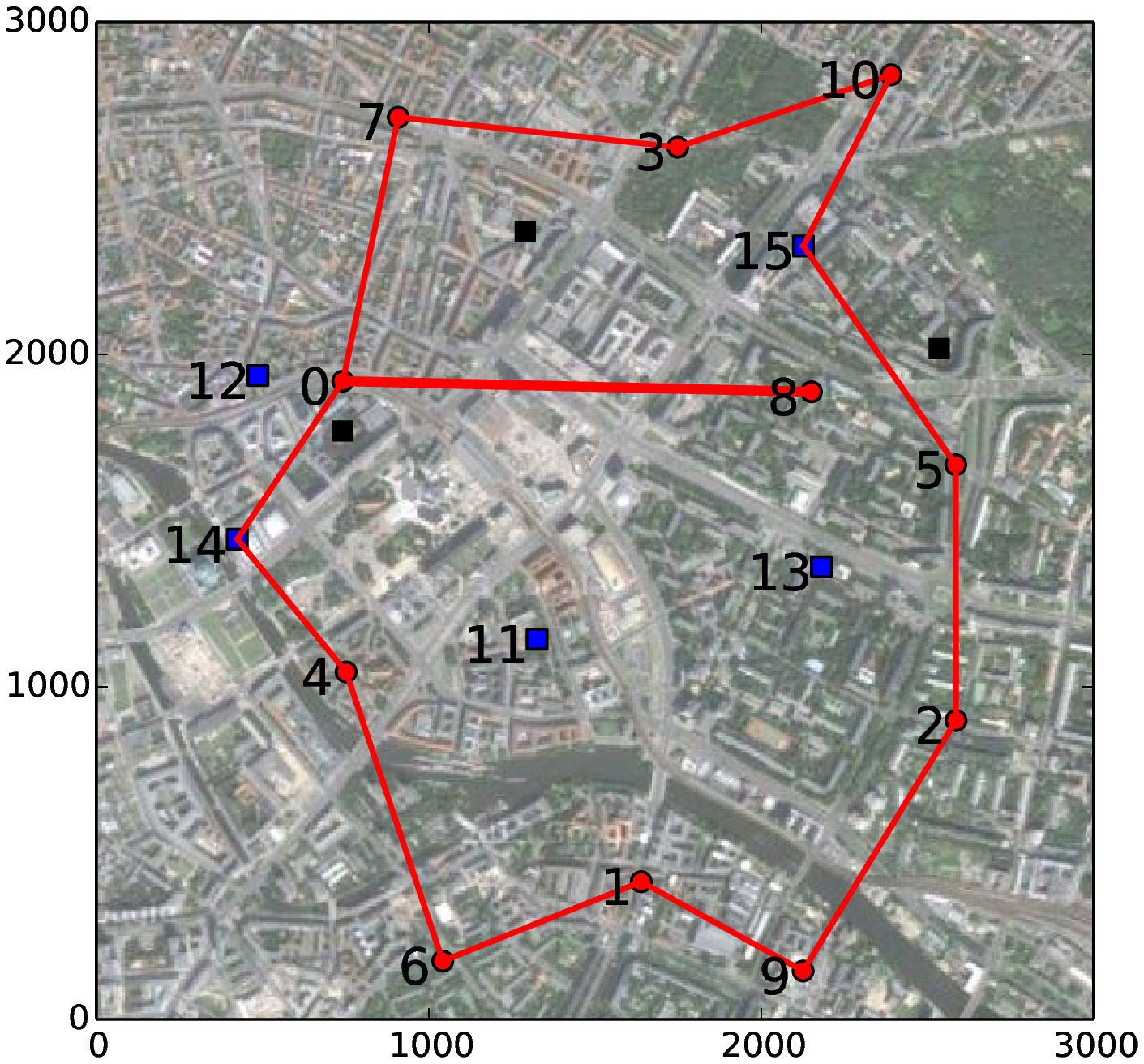}
                \caption{2-connected}
                \label{fig:2-connected}
        \end{subfigure}
        \begin{subfigure}[b]{0.3\textwidth}
                \centering
                \includegraphics[trim = 1.2in 0.3in 1.2in 0.3in, clip, width=2in]{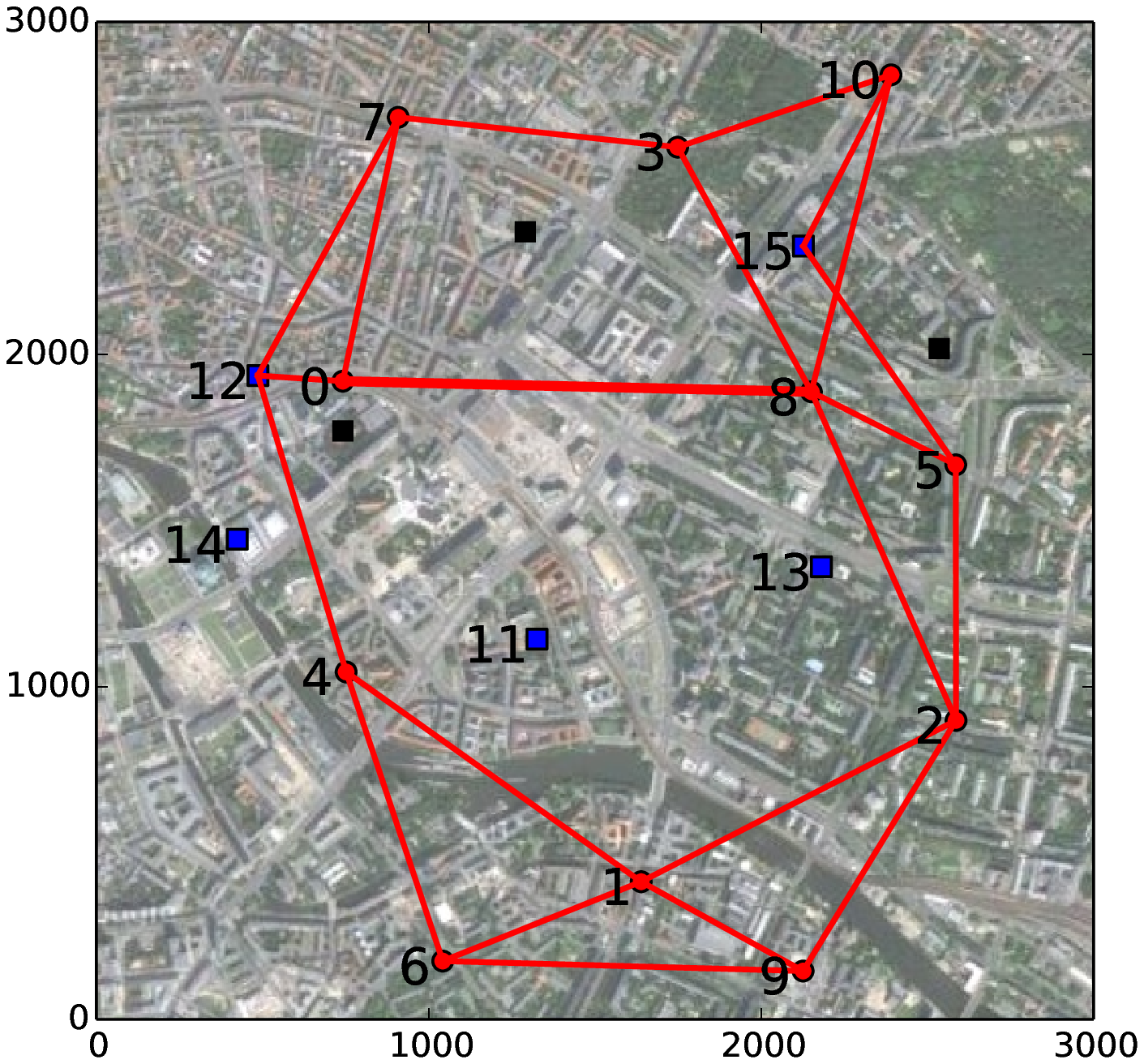}
                \caption{3-connected}
                \label{fig:3-connected}
        \end{subfigure}
        \caption{Optimal topologies with different connectivity for R1}\label{fig:connectivity}
\end{figure*}

\subsection{A case study}
Utilizing the CRBND model, we compute optimal solutions for R1, and then the impact of several parameters on optimal solutions is evaluated.
This case is studied from the cost perspective, i.e., $c_1 : c_4 = 10 :1$. The consideration from the reliability perspective, is presented in the next subsection.

In R1, there is an fiber link between nodes $0$ and $8$ represented by a thick line. The optimal topologies for $K=1,2,3$ are presented in Fig. \eqref{fig:connectivity}. Fig. \eqref{fig:connectivity}(a) shows an $1$-connected topology with $16$ FSO transceivers and $0$ mirrors. Fig. \eqref{fig:connectivity}(b) shows a $2$-connected topology consisting of $20$ FSO transceivers, $2$ mirrors and $2$ mirror stations.
Nodes $14$ and $15$ have mirrors to connect the FSO nodes that are not in line of sight. Note that paths that use the fiber link between nodes $0$ and $8$ are link-disjoint. Fig. \eqref{fig:connectivity}(c) shows a $3$-connected graph with $38$ FSO transceivers and $3$ mirrors. There are two mirrors deployed by node $12$. One is used to connect the nodes $0$ and $7$, the other one is used to connect nodes $0$ and $4$. Mirror node $14$ is not used in order to reduce the cost of leasing mirror nodes. Note that the link $\{0,12\}$ corresponds to $2$ FSO transceivers in node $0$.

In R2, we study the impact of $K$ and the number of available optical fibers on the optimal solution. We consider four cases regarding the number of optical fibers: $0,4,8,12$. For each case, we test $10$ instances and then take the average over the obtained results. The number of FSO transceivers, mirrors, leased mirror nodes for $K=1,2,3,4$ are plotted in Fig. \ref{fig:cost}.

In Fig. \ref{fig:fso}, we can see that as $K$ increases, the number of FSO transceivers increases dramatically, which is more significant than the increase of the number of mirrors in Fig. \ref{fig:mirrorNum} and the number of mirror nodes in Fig. \ref{fig:mirrorNodes}. This is because an FSO link needs two FSO transceivers and some mirror links also need FSO transceivers. However, as illustrated in Fig. \ref{fig:fso}, when more optical fibers are available, then the number of FSO transceivers grows slowly with the increase of $K$. By adding $12$ optical fibers, the number of FSO transceivers for $K=4$ decreases by $40$ compared to the case without optical fibers, thus, around $90\%$ of FSO transceivers are saved.

The increasing trend of the number of mirror in Fig. \ref{fig:mirrorNum} is similar to the trend of the number of used mirror nodes in Fig. \ref{fig:mirrorNodes}. However, the number of used mirror nodes increases slower than the number of mirrors, since a mirror node can accommodate more than one mirrors. The deployment of optical fibers slows down the increasing trend. The network reliability is improved when $K$ increases and more fiber links are added, which is trivial and we omit the related analysis.

\begin{figure}[!hbtp]
        \centering
        \begin{subfigure}[b]{0.32\textwidth}
                \centering
                \includegraphics[trim = 0.3in 0.3in 0.5in 0.3in, clip, width=2in]{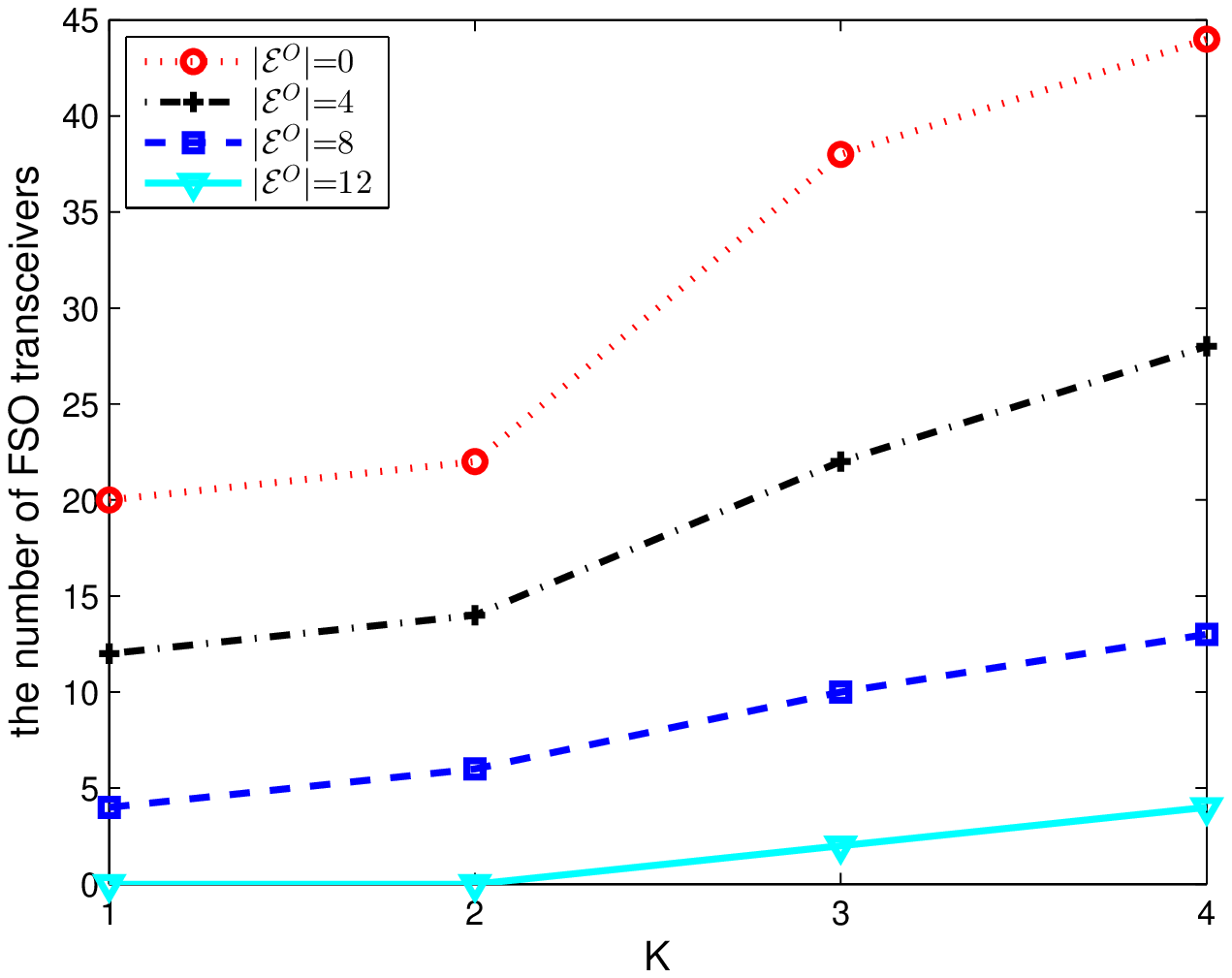}
                \caption{FSO transceivers}
                \label{fig:fso}
        \end{subfigure}
        \begin{subfigure}[b]{0.32\textwidth}
                \centering
                \includegraphics[trim = 0.3in 0.3in 0.5in 0.3in, clip, width=2in]{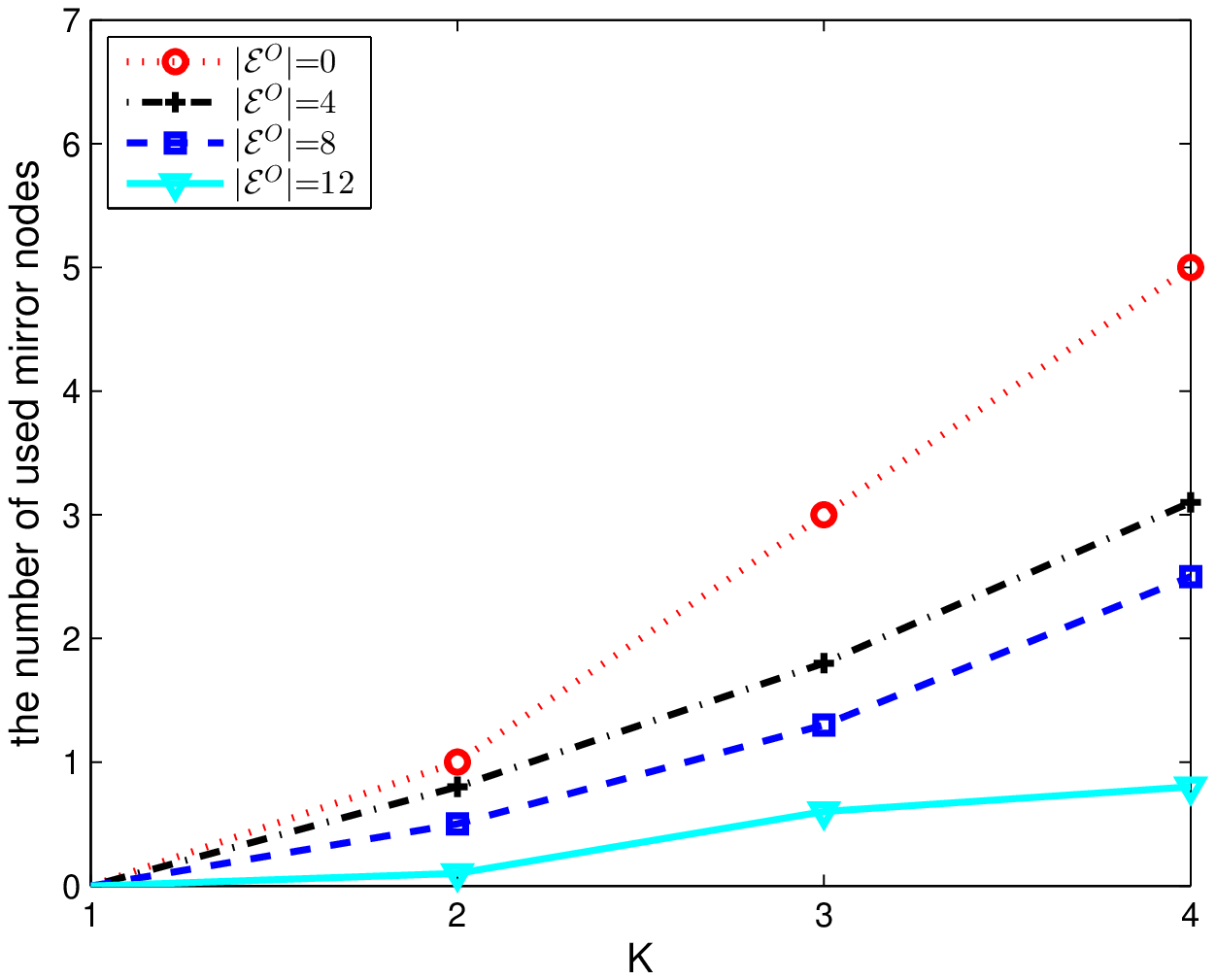}
                \caption{Mirror Stations}
                \label{fig:mirrorNodes}
        \end{subfigure}
        \begin{subfigure}[b]{0.32\textwidth}
                \centering
                \includegraphics[trim = 0.3in 0.3in 0.5in 0.3in, clip, width=2in]{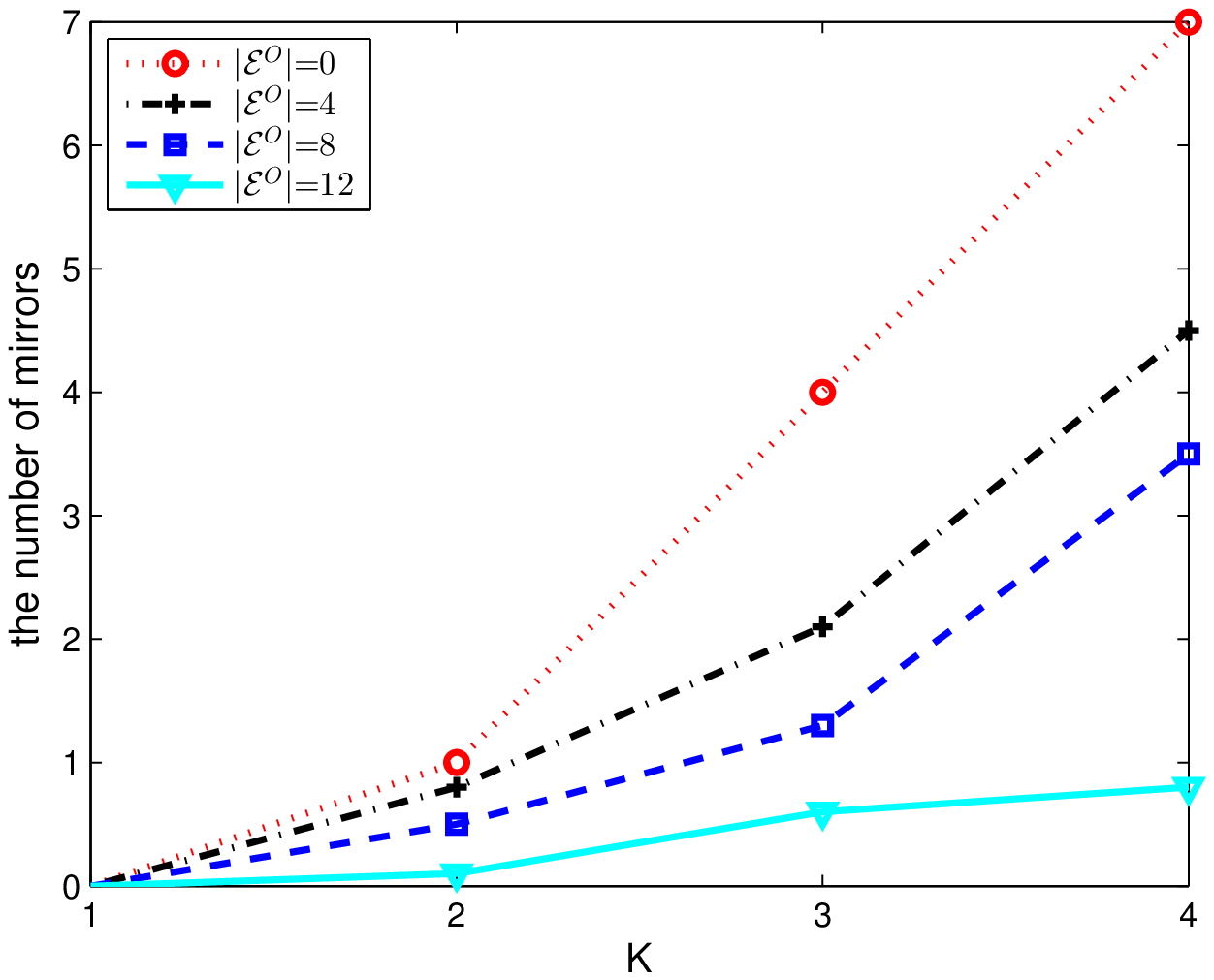}
                \caption{Mirrors}
                \label{fig:mirrorNum}
        \end{subfigure}
        \caption{The number of FSO transceivers, Mirrors, Mirrors Nodes v.s. Optical Fibers and $K$}\label{fig:cost}
\end{figure}

\subsection{Algorithm comparison}
In this section we make comparisons of the sequential computational approach (Algorithm 2) with the exact optimization model, CRBND, through a large number of experiments.

Table \ref{tab:comparison} shows optimal results and the running time of the two approaches for two settings, i.e., $c_4:c_1=10:1$ and $c_4:c_1=1:10$. The optimal results include the number of FSO transceivers ($F$), mirror nodes ($N$), mirrors ($M$) and the normalized reliability $R$. The normalized reliability is the total reliability of FSO and mirror links in the optimal solution divided by the total reliability of all potential FSO and mirror links. The entry ``$>$3h'' in the table denotes that the model or the algorithm cannot produce an optimal solution in 3 hours thus, ``--'' is inserted to the corresponding results.

We test a large variety of networks with different number of nodes. For each input network, we select in a random manner a $10\%$ of links to be links which are not in line of sight. Furthermore, we select $5\%$ pairs of FSO nodes between which optical fibers are used. All results are the average results of testing $10$ instances.

From Table \ref{tab:comparison}, we can see that the optimal solution corresponding to $c_4:c_1=10:1$ has higher reliability than that for the setting $c_4:c_1=1:10$.  However, the former one has higher cost in terms of the number of FSO transceivers, mirrors and leased mirror nodes. Under the same setting, i.e., the same $c_4:c_1$, the proposed sequential computation approach always obtains an optimal solution with higher reliability and thus needs higher cost. As we can see in the table, when the weight of reliability is higher, the gap of the reliability between the two approaches is smaller but the gap of the cost becomes bigger. Similarly, when the weight of cost is higher, the gap of cost becomes smaller but the gap of the reliability becomes bigger. For both approaches, the running time for $c_4:c_1=10:1$ is bigger than that for the $c_4:c_1=1:10$. Notably, the running time of the model CRBND for $c_4:c_1=10:1$ is much higher than that for $c_4:c_1=1:10$ while the running times for Algorithm 2 for two settings have relative small gap. This indicates that when the weight reliability increases, the running time of the model CRBND increases much faster than the sequential computation approach.

For small input networks, say $10$ nodes, the model CRBND can compute optimal solutions quickly. However, when the network size increases, even to a medium size, e.g., $26$ nodes, the model CRBND will be quite time consuming. Thus, we need to use the proposed sequential computation approach. For the impact of $K$ on the two methods, bigger $K$ takes longer running time for both methods. Even only increasing by one for $K$, the running times for both approaches increase a lot.

\begin{table}[!htbp]
\setlength{\tabcolsep}{3pt}
\center
\caption{Comparison of the cost the and running time}
\label{tab:comparison}
\begin{tabular}{c  c  c c c c c  c c  c c c c  c c  c c}
\toprule
\multicolumn{3}{c}{\multirow{2}{*}{Input}} & &  \multicolumn{6}{c}{CRBND ($c_4:c_1$)} & &  \multicolumn{6}{c}{Algorithm 2 ($c_4:c_1$)} \\
\cline{6-9} \cline{13-16}
\multicolumn{1}{c}{} & \multicolumn{1}{c}{}  & \multicolumn{1}{c}{}& &  \multicolumn{3}{c}{10:1} & \multicolumn{3}{c}{1:10} & & \multicolumn{3}{c}{10:1} & \multicolumn{3}{c}{1:10}\\
\cline{1-3} \cline{5-10} \cline{12-17}
$|\mathcal V|$ & $|\mathcal E|$ & $K$ & &  F,M,N & R &time(s) & F,M,N & R &time(s) & & F,M,N & R &time(s)  & F,M,N & R &time(s)\\
\hline
\multirow{2}{*}{10} &\multirow{2}{*}{74} &2  &  &10,5,5 &0.43 &16   &10,0,0 &0.12 &3  &   &12,4,3 &0.49 &3  &11,0,0 &0.16 &2\\
                                        &  &3 &  &17,5,5 &0.53 &23    &16,3,2 &0.33 &6  &  &20,5,4 &0.56 &5   &17,3,2 &0.36 &3 \\
\hline
\multirow{2}{*}{16} &\multirow{2}{*}{100} &2 &  &18,6,6 &0.40 &243   &12,1,1 &0.08 &65 &  &24,4,3 &0.36 &84   &14,0,0 &0.13 &50\\
                                          & &3  & &27,7,6 &0.50 &325   &16,1,1 &0.14 &98  &  &35,4,2 &0.53 &102   &22,4,2 &0.21 &87\\
\hline
\multirow{2}{*}{20} &\multirow{2}{*}{186} &2 &  &26,4,2 &0.31 &1568    &18,2,1 &0.06 &896 &  &40,7,4 &0.32 &876   &27,6,5 &0.18  &546\\
                                       &  &3  &  &34,6,4 &0.38 &2285   &22,3,2 &0.11 &1330  &  &48,7,6 &0.41 &1540   &30,5,2 &0.19 &959\\
\hline
\multirow{2}{*}{26} &\multirow{2}{*}{318} &2 &   &36,8,6   &0.25 &6578   &28,4,3 &0.05 &3674 &   &55,4,3 &0.27 &2425   &35,3,2 &0.12 &1034\\
                                       &  &3 &   &45,10,8  &0.34 &8854   &36,5,5 &0.09  &5578 &  &61,7,3 &0.38 &5789  &40,5,2 &0.15 &2578\\
\hline
\multirow{2}{*}{30} &\multirow{2}{*}{405} &2 &  &-- &-- &$>$ 3h     &-- &-- &$>$ 3h  &  &68,8,6 &0.23 &8764  &47,0,0 &0.08 &3983\\
                                          & &3  & &-- &-- &$>$ 3h   &-- &-- &$>$ 3h  &  &-- &-- &$>$ 3h      &54,6,4 &0.12 &7645\\
\bottomrule
\end{tabular}
\hfill
\end{table}

The benefits of using Algorithm 1 in Algorithm 2 on the running time are studied.
The tested networks are randomly generated with the following parameters $c_4:c_1 = 1:1$ and $K=2$. For networks with the same number of nodes, $10$ instances are tested and average results are presented. The running times for Algorithm 2 with Algorithm 1 (T1) and the running times for Algorithm 2 without Algorithm 1 (T2) are shown in Table \ref{tab:algorithm1}. We see that Algorithm 1 helps a lot to reduce the running time of Algorithm 2.

\begin{table}[!htbp]
\setlength{\tabcolsep}{17pt}
\center
\caption{Comparison of the running time of Algorithm 2 with and without Algorithm 1}
\label{tab:algorithm1}
\begin{tabular}{|c|c|c|c|c|c|c|}
\hline
         & 10 & 16 & 20 & 26 & 30\\
\hline
T1 (s)  & 4  &60  &689  &1123 &4456\\
\hline
T2 (s)  & 8  &115  &1764  &3565  &$>$3h\\
\hline
\end{tabular}
\hfill
\end{table}

\section{Conclusion} \label{sec:conclusion}
We have studied optimization formulation and solution approaches for designing FSO networks, taking into account cost and reliability as performance metrics. Network survivability is modelled by means of the $K$-connectivity requirement between FSO nodes. From an optimization viewpoint, the problem represents an unconventional setup within the domain of graph optimization, because of the possibility of using mirrors with distance consideration of such connections. Integer programming models based on the notions of network flow and path generation have been developed. The latter has been utilized to derive a sequential algorithm that, by the performance results, enables near-optimal solutions with much smaller computational effort in comparison to applying integer programming to the flow-based model. Moreover, the results illustrate the trade-off between cost and reliability, in particular for large-scale scenarios. Thus the study provides new insights into deploying cost-effective and high-performance networks with the FSO technology.

\bibliographystyle{IEEEtran}
\bibliography{ref}

\begin{thebibliography}{10}
\providecommand{\url}[1]{#1}
\csname url@samestyle\endcsname
\providecommand{\newblock}{\relax}
\providecommand{\bibinfo}[2]{#2}
\providecommand{\BIBentrySTDinterwordspacing}{\spaceskip=0pt\relax}
\providecommand{\BIBentryALTinterwordstretchfactor}{4}
\providecommand{\BIBentryALTinterwordspacing}{\spaceskip=\fontdimen2\font plus
\BIBentryALTinterwordstretchfactor\fontdimen3\font minus
  \fontdimen4\font\relax}
\providecommand{\BIBforeignlanguage}[2]{{%
\expandafter\ifx\csname l@#1\endcsname\relax
\typeout{** WARNING: IEEEtran.bst: No hyphenation pattern has been}%
\typeout{** loaded for the language `#1'. Using the pattern for}%
\typeout{** the default language instead.}%
\else
\language=\csname l@#1\endcsname
\fi
#2}}
\providecommand{\BIBdecl}{\relax}
\BIBdecl

\bibitem{Ford2013}
R.~Ford, C.~Kim, and S.~Rangan, ``{Opportunistic third-party backhaul for
  cellular wireless networks},'' \emph{ArXiv e-prints}, 2013.

\bibitem{Tipmongkolsilp2011}
O.~Tipmongkolsilp, S.~Zaghloul, and A.~Jukan, ``{The evolution of cellular
  backhaul technologies: current issues and future trends},'' \emph{IEEE
  Communications Surveys Tutorials}, vol.~13, no.~1, pp. 97--113, 2011.

\bibitem{Frey99}
T.~Frey, ``{The effects of the atmosphere and weather on the performance of a
  mm-wave communication link},'' \emph{Applied Microwave and Wireless}, pp.
  76--80, 1999.

\bibitem{Chia2009}
S.~Chia, M.~Gasparroni, and P.~Brick, ``{The next challenge for cellular
  networks: backhaul},'' \emph{IEEE Microwave Magazine}, vol.~10, no.~5, pp.
  54--66, 2009.

\bibitem{Demers2011}
\BIBentryALTinterwordspacing
F.~Demers, H.~Yanikomeroglu, and M.~St-Hilaire, ``{A survey of opportunities
  for free space optics in next generation cellular networks},'' in \emph{the
  9th Annual Communication Networks and Services Research Conference (CNSR)},
  2011, pp. 210--216. [Online]. Available:
  \url{http://dx.doi.org/10.1109/CNSR.2011.38}
\BIBentrySTDinterwordspacing

\bibitem{Chan06}
V.~Chan, ``{Free-space optical communications},'' \emph{IEEE/OSA Journal of
  Lightwave Technology}, vol.~24, no.~12, pp. 4750--4762, 2006.

\bibitem{Refai2006}
H.~H. Refai, J.~J. Sluss, H.~H. Refai, and M.~Atiquzzaman, ``Finding
  failure-disjoint paths for path diversity protection in communication
  networks,'' \emph{Optical Engineering}, vol.~45, no.~2, p. 025003, 2006.

\bibitem{Ghassemlooy2012}
Z.~Ghassemlooy, W.~Popoola, and S.~Rajbhandari, \emph{Optical Wireless
  Communications: System and Channel Modelling with MATLAB}.\hskip 1em plus
  0.5em minus 0.4em\relax CRC Press, 2012.

\bibitem{Kerivin2005}
H.~Kerivin and A.~Mahjoub, ``Design of survivable networks: A survey.''
  \emph{Networks}, vol.~46, no.~1, pp. 1--21, 2005.

\bibitem{Khandekar2013}
\BIBentryALTinterwordspacing
R.~Khandekar, G.~Kortsarz, and Z.~Nutov, ``On some network design problems with
  degree constraints,'' \emph{Journal of Computer and System Sciences},
  vol.~79, no.~5, pp. 725--736, 2013. [Online]. Available:
  \url{http://dx.doi.org/10.1016/j.jcss.2013.01.019}
\BIBentrySTDinterwordspacing

\bibitem{Bendali2010}
F.~Bendali, I.~Diarrassouba, A.~Mahjoub, and J.~Mailfert, ``The edge-disjoint
  3-hop-constrained paths polytope,'' \emph{Discrete Optimization}, vol.~7,
  no.~4, pp. 222 -- 233, 2010.

\bibitem{Botton2013}
Q.~Botton, B.~Fortz, L.~Gouveia, and M.~Poss, ``Benders decomposition for the
  hop-constrained survivable network design problem,'' \emph{INFORMS Journal on
  Computing}, vol.~25, no.~1, pp. 13--26, 2013.

\bibitem{Zotkiewicz2010}
M.~Zotkiewicz, W.~Ben-Ameur, and M.~Pioro, ``Finding failure-disjoint paths for
  path diversity protection in communication networks,'' \emph{IEEE
  Communications Letters}, vol.~14, no.~8, pp. 776--778, August 2010.

\bibitem{Borah2012}
\BIBentryALTinterwordspacing
D.~Borah, A.~Boucouvalas, C.~Davis, S.~Hranilovic, and K.~Yiannopoulos,
  ``\BIBforeignlanguage{English}{{A review of communication-oriented optical
  wireless systems}},'' \emph{\BIBforeignlanguage{English}{EURASIP Journal on
  Wireless Communications and Networking}}, vol. 2012, no.~1, pp. 1--28, 2012.
  [Online]. Available: \url{http://dx.doi.org/10.1186/1687-1499-2012-91}
\BIBentrySTDinterwordspacing

\bibitem{Llorca2004}
J.~Llorca, A.~Desai, and S.~Milner, ``{Obscuration minimization in dynamic free
  space optical networks through topology control},'' in \emph{IEEE MILCOM},
  vol.~3, 2004, pp. 1247--1253 Vol. 3.

\bibitem{Zhuang2004}
J.~Zhuang, M.~Casey, S.~Milner, S.~Gabriel, and G.~Baecher, ``{Multi-objective
  optimization techniques in topology control of free space optical
  networks},'' in \emph{IEEE MILCOM}, vol.~1, 2004, pp. 430--435.

\bibitem{Cao2008}
X.~Cao, ``{An integer linear programming approach for topology design in OWC
  networks},'' in \emph{IEEE GLOBECOM Workshops}, 2008, pp. 1--5.

\bibitem{Abhish2007}
\BIBentryALTinterwordspacing
A.~Kashyap, K.~Lee, M.~Kalantari, S.~Khuller, and M.~Shayman, ``Integrated
  topology control and routing in wireless optical mesh networks,''
  \emph{Computer Networks}, vol.~51, no.~15, pp. 4237 -- 4251, 2007. [Online].
  Available:
  \url{http://www.sciencedirect.com/science/article/pii/S138912860700151X}
\BIBentrySTDinterwordspacing

\bibitem{Son2010}
I.~Son and S.~Mao, ``{Design and optimization of a tiered wireless access
  network},'' in \emph{IEEE INFOCOM}, 2010, pp. 1--9.

\bibitem{Zhou2013}
H.~Zhou, A.~Babaei, S.~Mao, and P.~Agrawal, ``Algebraic connectivity of degree
  constrained spanning trees for fso networks,'' in \emph{IEEE ICC}, June 2013,
  pp. 5991--5996.

\bibitem{Ouveysi2010}
\BIBentryALTinterwordspacing
I.~Ouveysi, F.~Shu, W.~Chen, G.~Shen, and M.~Zukerman, ``Topology and routing
  optimization for congestion minimization in optical wireless networks,''
  \emph{Optical Switching and Networking}, vol.~7, no.~3, pp. 95 -- 107, 2010.
  [Online]. Available:
  \url{http://www.sciencedirect.com/science/article/pii/S1573427710000068}
\BIBentrySTDinterwordspacing

\bibitem{pioro2004}
M.~Pi\'{o}ro and D.~Medhi, \emph{Routing, Flow, and Capacity Design in
  Communication and Computer Networks}.\hskip 1em plus 0.5em minus 0.4em\relax
  Morgan Kaufmann, 2004.

\bibitem{Son2014}
\BIBentryALTinterwordspacing
I.~K. Son, S.~Mao, and S.~K. Das, ``On joint topology design and load balancing
  in free-space optical networks,'' \emph{Optical Switching and Networking},
  vol. 11, Part A, pp. 92 -- 104, 2014. [Online]. Available:
  \url{http://www.sciencedirect.com/science/article/pii/S1573427713000507}
\BIBentrySTDinterwordspacing

\bibitem{Zhu2002}
X.~Zhu and J.~Kahn, ``Free-space optical communication through atmospheric
  turbulence channels,'' \emph{IEEE Transactions on Communications}, vol.~50,
  no.~8, pp. 1293--1300, 2002.

\bibitem{Garey1980}
M.~R. Garey and D.~S. Johnson, \emph{Computers and intractability: A guide to
  the theory of NP-completeness}.\hskip 1em plus 0.5em minus 0.4em\relax
  American Mathematical Society, 1980.

\bibitem{Gurobi}
``{Gurobi Optimizer},'' http://www.gurobi.com.

\bibitem{MOMENTUM}
``Momentum project, 2003 (updated in 2005),'' http://momentum.zib.de.

\end{thebibliography}

\end{document}